\documentclass[10pt,a4]{article}
\usepackage{amsmath,amsfonts,amssymb,amsthm,relsize,graphicx,calc,array,hyperref,fullpage}
\usepackage{mathtools}
\mathtoolsset{showonlyrefs=true}
\usepackage{authblk}

\usepackage{color}

\newcolumntype{C}{>{\(}p{3mm}<{\)}}

 \usepackage[final,mono]{optional}

 \opt{final}
 {
 }
 \opt{xref}
 {
 }

\usepackage{enumitem}
\setlist[enumerate,1]{label={\upshape(\roman*)}, align=left, widest=iii, leftmargin=*}
\setlist[enumerate,2]{label={\upshape(\alph*)}, align=left, widest=a, leftmargin=*}

\makeatletter
\def\blfootnote{\xdef\@thefnmark{}\@footnotetext}
\makeatother

\newtheorem{thm}{Theorem}[section]
\newtheorem{lem}{Lemma}[section]
\newtheorem{cor}{Corollary}[section]
\newtheorem{prop}{Proposition}[section]
\theoremstyle{definition} 
\newtheorem{exa}{Example}[section]
\newtheorem{exas}{Examples}[section]
\newtheorem{defn}{Definition}[section]
\theoremstyle{remark} 
\newtheorem{rem}{Remark}[section]
\newtheorem{rems}{Remarks}[section]

\newcommand{\uplp}{{\upshape (}}
\newcommand{\uprp}{{\upshape )}}
\renewcommand{\le}{\leqslant}
\renewcommand{\ge}{\geqslant}
\renewcommand{\tilde}{\widetilde}

\newcommand{\omax}{\omega_{\text{max}}}
\newcommand{\kmax}{\kappa_{\text{max}}}

\title{Darboux dressing and undressing for the\\ ultradiscrete KdV equation}
\author[1]{Jonathan {J.C.} Nimmo${^\dagger}$}
\author[1]{Claire {R.}  Gilson}
\author[2]{Ralph Willox\thanks{willox@ms.u-tokyo.ac.jp}}
\affil[1]{School of Mathematics and Statistics, University of Glasgow, {Glasgow G12 8SQ,} UK.}
\affil[2]{Graduate School of Mathematical Sciences, the University of Tokyo,\break 3-8-1 Komaba, Meguro-ku, 153-8914 Tokyo, Japan.}
\date{}

\begin{document}
\maketitle
\begin{abstract}
{We solve the direct scattering problem for the ultradiscrete Korteweg de Vries (udKdV) equation, over $\mathbb R$ for any potential with compact (finite) support, by explicitly constructing bound state and non-bound state eigenfunctions. We then show how to reconstruct the potential in the scattering problem at any time, using an ultradiscrete analogue of {a} Darboux transformation. This is achieved by obtaining data uniquely characterising the soliton content and the `background' from the initial potential {by Darboux transformation}.} 
\end{abstract}
\blfootnote{{$^\dagger$ Sadly Jon Nimmo passed away on 20th June 2017}}

\section{Introduction}
The ultradiscrete KdV equation was first introduced by Takahashi and Satsuma in 1990 under the name soliton cellular automaton \cite{MR1082435}. A state in this system at time step $t$ is a sequence $(U^t_i)_{i\in\mathbb Z}$, with all but a finite number of terms non-zero. The state at time $t+1$ is determined by the update rule \eqref{eq:update rule}. In Takahashi and Satsuma's original formulation, $U^t_i$ takes binary values only, that is, cell values are restricted to $\{0,1\}$. The update rule then always gives another binary sequence and the system has an interpretation, which is now more commonly used, as a so-called box and ball system (BBS) in which the value 0 corresponds to an empty box and a 1 to a box occupied by a ball \cite{MR1192582}. Over the past decade it has become clear that the BBS and the mathematical tools used to describe its properties, are intimately related to many topics in mathematical physics such as Yang-Baxter maps, crystal base theory in quantum groups and tropical geometry, to name but a few  (cf. \cite{Inouereview} for an exhaustive review). 

At a more fundamental level, it has been known from the earliest papers that, in the BBS, all states evolve as $t\to\pm\infty$ into a finite number of blocks of consecutive 1s separated by blocks of 0s, where a block of $k$ consecutive 1s translating at speed $k$ is interpreted as a soliton of mass $k$. The two asymptotic states have the same block structure, with phase-shifts, and the evolution can be thought to represent interacting solitons. The initial value problem for the BBS was solved relatively recently, by introducing combinatorial quantities that play the role of action-angle variables and that are related to the Kerov-Kirillov-Reshetikhin bijection \cite{Takagi2005}, \cite{Inouereview} (or \cite{JIS} for an alternative, elementary, construction), as well as through a direct and explicit construction of the general $N$-soliton state that results from arbitrary initial conditions \cite{MR2451672}. 

By the mid-1990's \cite{TokiTaka1996} it had been realized the BBS can be seen to arise from the discrete KdV equation \eqref{eq:dkdv 1} via a limiting procedure called \emph{ultradiscretization} (also known in tropical mathematics as Maslov dequantisation, see \cite{MR2148995} for example). Hence the system is also referred to as the ultradiscrete KdV equation.  

In the last few years there has also been some interest in more general versions of the ultradiscrete KdV equation. These have the same update rule as the BBS but $U^t_i$ takes arbitrary integer \cite{MR2517833,MR2738130,MR2780417} or arbitrary real values \cite{MR2985392,MR3366309}. In this non-binary case, the most general solution describes the interaction of pair-wise interacting solitons of arbitrary positive mass, similar to the BBS solitons, overlayed on a simply evolving background. We will give more details of these solutions later. In particular, in \cite{MR2738130,MR2985392}, a procedure for solving the initial value problem for the ultradiscrete KdV equation by an inverse scattering method was described. This method bears a striking resemblance to the classical IST scheme for the (continuous) KdV equation \cite{DeiftTrub1979} and the action-angle variables that appear in it play exactly the same role as those in the continuous case. A key part of this procedure, and the most difficult part, is the use of Darboux transformations to ``undress'' all of the solitons at time $t=0$ one by one, and to determine their defining parameters, mass and phase. This parameter data, and the background that remains, evolve very simply, in fact linearly, in time and the solution at time $t$ can be reconstructed by means of a sequence of B\"acklund transformations \cite{MR2545618}.

The main results of this paper are a general and explicit expression for the special undressing eigenfunction that removes the heaviest soliton in a given state. It is also shown that the B\"acklund transformation used to reconstruct the solution comes from exactly the same Darboux transformation but with a different type of eigenfunction, one that does not correspond to a bound state. {This eigenfunction is constructed explicitly as well.} The paper is organised as follows: in Section \ref{sec:dKdV and udKdV} we state some standard results on the discrete and ultradiscrete KdV equation including the ultradiscrete linear system (Lax pair) for the ultradiscrete KdV equation.  Section~\ref{sec:sol prop} gives some details of the general solution of ultradiscrete KdV in the case that the solution $U^t_i\in\mathbb R$ and we give criteria which characterise solutions with and without soliton content. An alternative version of the update rule {for udKdV} is described in Section \ref{sec:update}. This alternative method enables us to deduce some simple properties of the system, including the conservation of total mass, and might have more general interest. In Section \ref{sec:conserved} we establish another two (as far as we know) new conserved quantities which will play a vital role in finding the dressing and undressing transformations. {In Section \ref{sec:darboux} we review some properties of Darboux transformations for the discrete KdV and introduce an ultradiscrete analogue.} In Section \ref{sec:line sol} we obtain a two parameter family of solutions of the ultradiscrete linear system expressed as the maximum of two basis functions {and we use such solutions to construct dressing Darboux transformations. In Section \ref{sec:bound states} we find the special choices of parameters for which the \emph{minimum} of the basis solutions is also a solution and we show how this solution defines an undressing Darboux transformation.} Finally in Section \ref{sec:cauchy} we {apply the dressing and undressing Darboux transformations to the Cauchy problem for the udKdV equation on $\mathbb R$ and we give a detailed example of the calculations involved.}

\section{Discrete and ultradiscrete KdV equations}\label{sec:dKdV and udKdV}
We first give a summary of some known results concerning the discrete and ultradiscrete KdV equations. 
\subsection{Discrete KdV}
The discrete KdV equation (dKdV) \cite{MR0460934} is the integrable partial difference equation 
\begin{equation}\label{eq:dkdv 1}
  \frac1{u^{t+1}_{i+1}}+\delta u^{t+1}_i=\frac1{u^{t}_{i}}+\delta u^{t}_{i+1},
\end{equation}
where $i,t\in\mathbb Z$ and $u^t_i\in(0,\infty)$ and where $\delta$ is a real constant not equal to 0 or 1 (values for which there is no continuum limit to the KdV equation). This equation arises directly from the system
\begin{equation}\label{eq:dkdv system}
  u^{t+1}_{i}=\frac{v^t_i}{1+\delta u^t_iv^t_i},\quad v^t_{i+1}=u^t_i(1+\delta u^t_iv^t_i),
\end{equation}
(a reduction of the Hirota-Miwa equation  \cite{MR2481229}) when $v^t_i$ is eliminated. Alternatively, assuming boundary conditions $u^t_i\to1$ and $v^t_i\to1/(1-\delta)$ as $i\to\pm\infty$ for any $t$, one may express $v^t_i$ as an infinite product in two ways
\begin{equation}\label{eq:v}
v^t_i=\frac1{1-\delta}\prod_{j<i}\frac{u^t_j}{u^{t+1}_j}\quad\text{or}\quad
v^t_i=\frac1{1-\delta}\prod_{j\ge i}\frac{u^{t+1}_j}{u^{t}_j}.
\end{equation}
From this, assuming convergence, it follows that
\begin{equation}\label{eq:cons u}
\prod_{i\in\mathbb Z} u^t_i\ 
\end{equation}
is independent of $t$.

Eliminating $v^t_i$ from \eqref{eq:dkdv system} using \eqref{eq:v} gives an alternative form of the dKdV equation
\begin{equation}\label{eq:dkdv 2}
\frac1{u^{t+1}_i}=\delta u^t_i+(1-\delta)\prod_{j<i}\frac{u^{t+1}_j}{u^{t}_j},
\end{equation}
which now defines an evolution in the positive $i$ direction, in which the values $\{u_j^t | j\leq i\}$ define $u_i^{t+1}$. Notice that equation \eqref{eq:dkdv 1} is invariant under the changes $(u, \delta ; i, t) \to (u, -\delta ; t, i)$ and $(u, \delta ; i, t) \to (1/u, 1/\delta ; i, -t)$ and that $\delta$ can therefore be restricted to $\delta\in(0,1)$ without loss of generality. From the alternative form \eqref{eq:dkdv 2} it follows that for such $\delta$, positive initial values $\{u_i^0\}$ always result in positive values for $u_i^t$, for all $t$ and $i$.

The dKdV equation has bilinear form \cite{MR0460934} 
\begin{equation}\label{eq:H dKdV}
\tau_{i+1}^{t+1}\tau_{i}^{t-1}=(1-\delta)\tau_{i+1}^{t}\tau_{i}^{t}+\delta\tau_{i}^{t+1}\tau_{i+1}^{t-1},
\end{equation}
where
\begin{equation}\label{eq:u tau}
  u_{i}^t=\frac{\tau_{i}^{t+1}\tau_{i+1}^{t}}{\tau_{i+1}^{t+1}\tau_{i}^{t}},
\end{equation}
and Lax pair
\begin{gather}
  \label{eq:dlin1}
  \phi_{i-1}^{t}-\frac2{1+\delta}\left(\frac1{u_{i}^{t}}+\delta u_{i-1}^{t}\right)\phi_{i}^{t}+(1-\lambda^2)\phi_{i+1}^{t}=0\\
  \label{eq:dlin2}
  \beta\phi_{i}^{t+1}=\phi_{i-1}^{t}+(\beta-1)u_{i-1}^{t}\phi_{i}^{t},
\end{gather} 
where $\delta\in(0,1)$ and $\beta:=(1-\delta)/(1+\delta)$ also lies in the interval $(0,1)$. 

This linear problem is not self-adjoint. Its adjoint is 
\begin{gather}
\label{eq:dlinad1}
\psi_{i+1}^t-\frac{2}{1+\delta} {\left(\frac{1}{u^{t-1}_{i-1}}+\delta u^{t-1}_{i}\right)} \psi^t_{i}+(1-\lambda^2)\psi^t_{i-1}=0\\
\label{eq:dlinad2}
\beta\psi^{t-1}_i=\psi^t_{i+1}+(\beta-1)u^{t-1}_{i}\psi^t_i.
\end{gather}
However, this adjoint is gauge equivalent to \eqref{eq:dlin1}, \eqref{eq:dlin2} since any solution $\phi^t_i$ of the linear problem \eqref{eq:dlin1} and \eqref{eq:dlin2} gives a solution 
\begin{equation}\label{eq:adjoint}
  \psi_i^t:=(1-\lambda^2)^i(1-\lambda^2/\beta^2)^{-t}\phi_i^t,
\end{equation} 
of \eqref{eq:dlinad1} and \eqref{eq:dlinad2}.

The squared eigenfunction potential $\Omega^t_i(\phi,\psi)=\Omega(\phi^t_i,\psi_i^t)$ is defined by the compatible difference equations
\begin{equation}\label{eq:Omega}
  \Omega_{i}^t-\Omega_{i-1}^t=\phi^t_i\psi^t_{i-1},\quad \Omega_{i}^{t+1}-\Omega_{i}^t=-\frac1\beta\phi^t_i\psi^{t+1}_{i}.
\end{equation}
In general, under the assumption that $u^t_i\to1$ as $i\to\pm\infty$,
\begin{equation}\label{eq:infty}
\phi^t_i\sim a(1+\lambda)^{-i}(1+\lambda/\beta)^t+b(1-\lambda)^{-i}(1-\lambda/\beta)^t,
\end{equation}
as $i\to-\infty$ and
\begin{equation}\label{eq:-infty}
\phi^t_i\sim c(1+\lambda)^{-i}(1+\lambda/\beta)^t+d(1-\lambda)^{-i}(1-\lambda/\beta)^t,
\end{equation}
as $i\to+\infty$, for some constants $a$, $b$, $c$ and $d$.

From here on we shall always assume that $0<\lambda<\beta<1$ (which is known to correspond to right-going solitons for the dKdV equation with $u_i^t\geq 1$). In this case,
\begin{equation}
  (1+\lambda)^{-i}\to
  \begin{cases}
    \infty&\text{as }i\to-\infty\\    
    0&\text{as }i\to+\infty
  \end{cases},\quad
  (1-\lambda)^{-i}\to
  \begin{cases}
    0&\text{as }i\to-\infty\\    
    \infty&\text{as }i\to+\infty
  \end{cases},
\end{equation}
and so if the constants $a$ and $d$ are both zero, $\phi^t_i\to0$ as $i\to\pm\infty$. {Then}
\begin{equation}
\phi^t_i\psi^t_{i-1}\sim b^2\frac{(1+\lambda)^i(1-\lambda/\beta)^t}{(1-\lambda)^i(1+\lambda/\beta)^{t}}\to0
\end{equation}
as $i\to-\infty$ and
\begin{equation}
\phi^t_i\psi^t_{i-1}\sim c^2\frac{(1-\lambda)^i(1+\lambda/\beta)^t}{(1+\lambda)^i(1-\lambda/\beta)^{t}}\to0
\end{equation}
as $i\to+\infty$. In such cases we can express $\Omega_t^i$ as a semi-infinite sum
\begin{equation}
  \Omega_{i}^t=\sum_{j\le i} \phi^t_{j}\psi^t_{j-1},
\end{equation}
and then define the norm, 
\begin{equation}\label{def: dnorm}
\Vert\phi^t\Vert:=\left\lvert\lim_{i\to+\infty}\Omega_i^t\right\rvert=\left\lvert\sum_{i\in\mathbb Z}\phi^t_i\psi^t_{i-1}\right\rvert,
\end{equation}
which, given the asymptotics of $\phi_i^t\psi_{i-1}^t$, is finite, and $\phi^t_i$ is said to correspond to a bound state. We have seen, in fact, that $\phi^t_i\to0$ as $i\to\pm\infty$ is sufficient for $\phi_i^t$ to be a bound state.

Notice also that for $i\to\pm\infty$,
\begin{equation}
\phi^t_i\psi^{t+1}_{i}\sim \phi_i^t\psi_{i-1}^t\to0.
\end{equation}
Taking the limit as $i\to\pm\infty$ in the second equation in \eqref{eq:Omega} we see that 
\begin{equation}
  \Vert\phi^{t+1}\Vert=\Vert\phi^t\Vert=\Vert\phi\Vert,
\end{equation}
and so the norm of an eigenfunction for a bound state is $t$-independent.

\subsection{Ultradiscrete KdV}
When taking the ultradiscrete limit, one assumes that $u^t_i=O(\delta^{-U^t_i})$, for positive $\delta\approx0$, and keeping only the lowest order terms in \eqref{eq:dkdv 1} one obtains the naive ultradiscrete form of dKdV, 
\begin{equation}\label{eq:naive}
  \max(U^{t+1}_{i}-1,-U^{t+1}_{i+1})=\max(U^t_{i+1}-1,-U^{t}_{i}),
\end{equation}
for all $i,t\in\mathbb Z$. From now on, we shall only consider solutions with finite support, that is we assume that $U^t_i=0$ for $|i|$ sufficiently large. By the same limiting process, \eqref{eq:cons u} gives 
\begin{equation}\label{eq:cons U}
\sum_{i\in\mathbb Z}U^t_i\ \ \text{is independent of $t$.}
\end{equation} 

The naive form \eqref{eq:naive} cannot be used to determine the time evolution uniquely and instead one considers the ultradiscrete limit of \eqref{eq:dkdv 2} to obtain the \emph{update rule}
\begin{equation}\label{eq:update rule}
  U^{t+1}_{i}=\min\Big(1-U^{t}_i,\sum_{j<i}(U^t_j-U^{t+1}_j)\Big).
\end{equation}
Changing $t$ to $t-1$ and rewriting \eqref{eq:update rule} using the conserved quantity \eqref{eq:cons U} gives
\begin{align*}
  U^{t}_{i}+U^{t-1}_i&=\min\Big(1,\sum_{j\le i}U^{t-1}_j-\sum_{j<i}U^{t}_j\Big)\\
                     &=\min\Big(1,\sum_{j\ge i}U^{t}_j-\sum_{j>i}U^{t-1}_j\Big),
\end{align*}
which gives the \emph{downdate rule}
\begin{equation}\label{eq:downdate rule}
  U^{t-1}_{i}=\min\Big(1-U^{t}_i,\sum_{j>i}(U^t_j-U^{t-1}_j)\Big).
\end{equation}
Note that if $U^t_i$ satisfies \eqref{eq:update rule} then \eqref{eq:naive} is also satisfied using the associativity and commutativity of $\min$. 

The bilinear form \eqref{eq:H dKdV} has ultradiscrete limit ($\tau^t_i=O(\delta^{-T^t_i})$)
\begin{equation}\label{eq:H udKdV}
T_{i+1}^{t+1}+T_{i}^{t-1}=\max(T_{i+1}^{t}+T_{i}^{t},T_{i}^{t+1}+T_{i+1}^{t-1}-1),
\end{equation}
where 
\begin{equation}\label{eq:ud H sub}
U^t_i=T^t_{i+1}+T^{t+1}_i-T^t_{i}-T^{t+1}_{i+1}.
\end{equation}

We assume that there is a solution $\phi_i^t$ of \eqref{eq:dlin1},\eqref{eq:dlin2} that is positive for all $i,t$. Then, taking the ultradiscrete limit ($\phi^t_i=O(\delta^{-\Phi^t_i}),1-\lambda^2=O(\delta^\kappa),1-\lambda^2/\beta^2=O(\delta^\omega)$) on the dKdV Lax pair leads to 
\begin{gather}
\label{eq:lin1}
  \max(\Phi^{t}_{i+1}-\kappa,\Phi^{t}_{i-1})=\Phi^{t}_{i}+\max(U^t_{i-1}-1,-U^t_i),\\
\label{eq:lin2}  
 \max(\Phi^{t+1}_{i+1}-\kappa,\Phi^{t+1}_{i-1})=\Phi^{t+1}_{i}+\max(U^t_{i}-1,-U^t_{i-1}),\\
\label{eq:lin3}  
  \max(\Phi^{t+1}_{i+1},\Phi^{t}_{i+1}+U^t_i-1)=\Phi^t_i,\\
 \label{eq:lin4}
  \max(\Phi^{t}_{i}+\kappa-\omega,\Phi^{t+1}_{i}+U^t_i+\kappa-1)=\Phi^{t+1}_{i+1},
\end{gather}
where $\omega$ is a nonnegative free parameter (spectral parameter) and $\kappa=\min(1,\omega)$. We call this system the ultradiscrete linear system for udKdV  \cite{MR2738130}, where `linearity' has to be understood as over the $(\max,+)$ semi-field: if $\Phi_i^t$ and $\widehat\Phi_i^t$ satisfy the system then so do $\max(\Phi_i^t,\widehat\Phi_i^t)$ and $\Phi_i^t+\phi$, for arbitrary constant $\phi$. {The rationale for imposing four linear equations in the ultradiscrete case instead of the customary two (as in the discrete or continuous case) is explained in \cite{MR2738130}.}

\section{Description of the general solution of udKdV}\label{sec:sol prop}
It is well known that the general solution of the box and ball system (udKdV with $U^t_i\in\{0,1\}$ and finite support) consists of interacting solitons made up of strings of $k$ consecutive 1s, propagating at speed $k$. The most general solution in the case $U^t_i\in\mathbb R$ is more complicated but can still be described completely: it consists of interacting solitons parametrised by {positive parameters $\omega$, plus a ``background'' that moves with speed 1, which is the minimal speed in the system \cite{MR2985392,MR3366309} (see also Remark \ref{rem:separation}).}

Any pair of a positive real number $\omega$ and a real phase constant $\phi$ fully describes a soliton solution to udKdV. Its $T$-function can be obtained as the ultradiscrete limit of a dKdV soliton but can also be verified directly by substitution in \eqref{eq:H udKdV}. It is given by
\begin{equation}
  T^t_i=\max(0,\kappa(i-\phi)-\omega t)=\max(0,\kappa(i-\varphi^t)),
\end{equation}
where $\kappa=\min(1,\omega)$ (the wave number) and $\varphi^t=\phi+c t$ (a time-dependent phase) in which the wave speed $c=\max(1,\omega)$ is always at least 1 {(note that $\kappa c = \omega$)}. For $\omega\le1$, $\kappa=\omega$ and $c=1$ whereas for $\omega\ge1$, $\kappa=1$ and $c=\omega$. From \eqref{eq:ud H sub} the solution $U^t_i$ is expressed in terms of four copies of the $T$-function: $U^t_i=T^t_{i+1}+T^{t+1}_{i}-T^t_{i}-T^{t+1}_{i+1}$.

It is beneficial here to study this solution $U^t(x)=U^t_i$ with the discrete space variable $i$ replaced by a real-valued variable $x$ \cite{MR2985392}. The expression for $U^t(x)$ is written explicitly as
\begin{align}
  U^t(x)
  &=\begin{cases}
    \kappa(x+1-\varphi^t)&\varphi^t-1\le x<\varphi^t\\
    \kappa&\varphi^t\le x<\varphi^{t+1}-1\\
    -\kappa(x-\varphi^{t+1})&\varphi^{t+1}-1\le x<\varphi^{t+1}\\
    0&\text{otherwise.}
  \end{cases}
\end{align}
\begin{figure}[htbp]\centering
\includegraphics*[width=13cm]{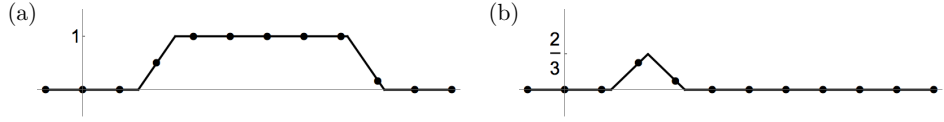}\vskip-.25cm
\caption{\label{fig:sol} Solitons plotted on $\mathbb R$ for $t=0$. In each plot the solid line is the graph of $U^t(x)$ and the points show the values of $U^t_i$ ($i\in\mathbb Z$). In (a) $\omega=17/3$, $\phi=5/2$ and in (b) $\omega=2/3$, $\phi=9/4$.}
\end{figure}

In Figure~\ref{fig:sol} (a) $\omega=17/3>1$ and so $\kappa=1$ and $c=17/3$. The solution at integer points is $U^0_i=\dots,0,0,1/2,1,1,1,1,1,1/6,0,\dots$. When considered in terms of the real variable $x$, the soliton simply translates at speed $c$ but when $c$ is non-integer, due to a stroboscopic effect, the soliton is not of fixed form on the integer lattice. In (b) $\omega=2/3\le1$ and so $\kappa=2/3$ and $c=1$. Hence $U^t_i=\dots,0, 0, 1/2, 1/6, 0,\dots$ which propagates without change at speed 1. 

In all cases, the area under the (real) curve $U^t(x)$ is $\kappa c=\omega$ and so we call $\omega$ the soliton mass. If we restrict back to integer $i$, each soliton has formula

\begin{equation}\label{eq:exp sol}
  U^t_i
  =\begin{cases}
    \kappa(1-\{\varphi^t\})&i=\lfloor \varphi^t\rfloor\\
    \kappa&\lfloor \varphi^t\rfloor+1\le i<\lceil \varphi^{t+1}\rceil-1\\
    \kappa(1-\{-\varphi^{t+1}\})&i=\lceil \varphi^{t+1}\rceil-1\\
    0&\text{otherwise,}
  \end{cases}
\end{equation}
where $\lfloor x\rfloor$ and $\lceil x\rceil$ denote, respectively, the floor and ceiling of a real number $x$ and $\{x\}:=x-\lfloor x\rfloor$ denotes the fractional part of $x$. Notice that the sum over all  $U_i^t$ always equals the soliton mass: 
\begin{align}
\sum_i U_i^t&=\kappa(1-\{\varphi^t\})+(\lceil\varphi^{t+1}\rceil-1-(\lfloor \varphi^t\rfloor+1))\kappa+\kappa(1-\{-\varphi^{t+1}\})\\
&=\kappa(-\varphi^t+\lfloor\varphi^t\rfloor+\lceil\varphi^{t+1}\rceil-\lfloor \varphi^t\rfloor-(-\varphi^{t+1})+\lfloor-\varphi^{t+1}\rfloor)\\
  &=\kappa(\varphi^{t+1}-\varphi^{t})=\kappa c = \omega, 
\end{align}
since for any $x$, $\lfloor x\rfloor+\lceil -x\rceil=0$.

In Figure~\ref{fig:sol ani}, (a) and (b) both show an animation over one period of a soliton with $\omega=5/3$ but with different phases. It is seen that each moves with (average) speed $c=5/3$, a translation of 5 space units in 3 time units, but at integer points, the {solutions depicted in (a) and (b) }never agree. 
\begin{figure}[htbp]\centering
\includegraphics*[width=15cm]{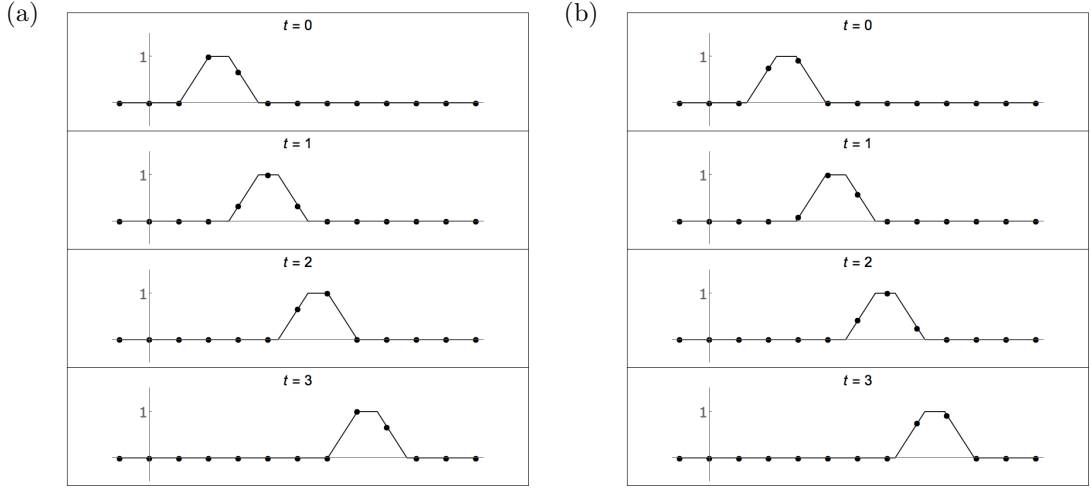}\vskip-.15cm
\caption{\label{fig:sol ani} Solitons plotted on $\mathbb R$ at times $t=0,1,2,3$. In (a) $\omega=5/3$, $\phi=2$ and in (b) $\omega=5/3$, $\phi=9/5$. In each case, the stroboscoptic effect is illustrated and the intial state is recovered, translated by 5 in 3 units of time, giving speed $5/3$.}
\end{figure}

Define the maximal local sum 
\begin{equation}\label{eq:V^t} 
V^t:=\max_i (U^t_i+U^t_{i+1}),
\end{equation} 
which is non-negative since $U^t_i=0$ for large enough $|i|$. The state $U^t_i=0$ for all $i,t$ is called \emph{trivial}. If $U^t_i$ is nontrivial but $V^t=0$ then we call $U^t_i$ \emph{background}. A background state translates with speed 1 in the direction of increasing $i$. Indeed, any state with $V^t\le1$ can be shown to translate at speed 1 ({see} Corollary~\ref{cor:cons}). This makes it impossible, by superficial examination of the asymptotic state $U^t_i$ at $t\to\infty$ alone, to distinguish background from solitons of mass $\omega\le1$. 

{In \cite{MR2517833}, Hirota has shown how to construct the $T$-function for a background solution to the ultradiscrete KdV equation. As will be {proved} in Corollary~\ref{cor:cons},  a background, say $B_i^t$, translates at speed 1 and so $B^t_i=B^0_{i-t}$. 
Now we observe that the Kronecker $\delta$-function can be written as
\begin{equation}
  \delta_{i,j}=\dfrac{1}{2} \big(|i+1-j|+|i-1-j|-2|i-j|\big),
\end{equation}
and so 
\begin{equation}
  B^t_i=\sum_{j\in\mathbb Z}\delta_{i-t,j}B^0_j=T^t_{i+1}+T^{t+1}_{i}-T^t_{i}-T^{t+1}_{i+1},
\end{equation}
where 
\begin{equation}\label{backgroundT}
  T^t_i=\dfrac{1}{2}\sum_{j\in\mathbb Z}|i-t-j|B^0_j,
\end{equation}
is the background $T$-function \cite{MR2517833}.
}

{
When $V^t>1$, in general, not much can be said about the $T$-functions at the level of the bilinear equation \eqref{eq:H udKdV}, besides their asymptotic behaviour in $i$. Let us define
\begin{equation}\label{Zdef}
Z_i^t= T_i^t - T_i^{t+1},
\end{equation}
which acts as a discrete potential for $U_i^t$: $U_i^t=\Delta Z_i^t = Z_{i+1}^t-Z_i^t$. Since $U_i^t$ has finite support it is clear {that asymptotically, for $|i|\gg 1$, $Z_i^t$  will take values $Z_+^t$ or $Z_-^t$ depending on the sign of $i$ but otherwise independent of $|i|$.} In fact, these two values do not depend on $t$ either. Replacing $T_i^{t+1}$ by $T_i^t-Z_i^t$ and $T_i^{t-1}$ by $T_i^t + Z_i^{t-1}$ in the ultradiscrete bilinear equation \eqref{eq:H udKdV}, we find in the asymptotic regime where 
$U_i^t=0~{\forall} |i|\gg1$ that
\begin{equation}
Z_\pm^{t-1} - Z_\pm^t = \max(0, Z_\pm^{t-1} - Z_\pm^t-1),
\end{equation}
and that $Z_\pm^t=Z_\pm^{t-1}$ for all $t$. The asymptotic values of this discrete potential, $Z_+$ and $Z_-$, are obviously related by
\begin{equation}
Z_+ - Z_- = \sum_{i\in\mathbb{Z}} U_i^t .
\end{equation}
Note that since $U_i^t$  has finite support, this last sum is actually a finite sum.}

{
For example, the asymptotics for the background $T$-function \eqref{backgroundT} is
\begin{equation}\label{BGZ}
Z_\pm = \lim_{i\to\pm\infty} \frac12 \sum_{j\in\mathbb{Z}} \big( |i-t-j| - |i-t-1-j| \big) B_j^0 = \pm \frac12 \sum_{j\in\mathbb{Z}} B_j^0,
\end{equation}
i.e. +1 or -1 times half the mass of the background.
}

{
Notice also that, by using the natural gauge freedom
\begin{equation}\label{Tgauge}
T_i^t \mapsto T_i^t + \alpha\, i + \beta\, t + \gamma\qquad (\alpha, \beta, \gamma\in\mathbb{R})
\end{equation}
we have in defining a $T$-function by \eqref{eq:H udKdV} or \eqref{eq:ud H sub} (both relations are invariant under such a transformation, but $Z_i^t\mapsto Z_i^t-\beta$), we can assign any value we choose to either $Z_+$ or $Z_-$.
}

{It was already mentioned in passing that the general solution of udKdV (with finite support) consists of a finite number of solitons of masses $0<\omega_1\le\cdots\le \omega_n$ plus a background. In the next few sections we shall develop a set of tools that will allow us to prove this statement and that yield an algorithm for obtaining analytic expressions for the $T$-functions for such general solutions to the udKdV equation.}

\section{Alternative version of the update and downdate rules}\label{sec:update}
In this section, we give alternative descriptions of the update and downdate rules. These allow us to prove some basic results on the udKdV evolution in an elementary way. They also have some computational advantage over the more usual formulae.

Let $a=(a_i)_{i\in\mathbb Z}$ denote a real sequence. It is assumed that $a$ is summable, that is, $\sum a:=\sum_{i\in\mathbb Z}a_i$ is finite and so necessarily $a_i\to0$ as $i\to\pm\infty$.  
\begin{defn}
Define $R_j:\mathbb R\to\mathbb R$, for $j\in\mathbb Z$, by
\begin{equation}
R_j(a_j)=\min(1,a_j),\ \ R_j(a_{j+1})=a_{j+1}+a_j-\min(1,a_{j}),\ \ R_j(a_i)=a_i\ (i\ne j,j+1).
\end{equation}
This definition is extended to real sequences in the natural way: $R_j(a)=(R_j(a_1),R_j(a_2),\dots)$.
\end{defn}
\begin{rem}
One may visualise $(a_i)$ as the contents of an array of cells labelled $i$ whose preferred capacities are each 1. Then the action of operator $R_j$ is to transfer the excess content of an overfull cell $j$ to the right neighbouring cell $j+1$ (even though it might already be overfull or become overfull). If cell $j$ is not overfull, $a_j\le1$, then $R_j$ acts as the identity. Note also that for any $j$, $R_j(a_{j})+R_j(a_{j+1})=a_{j}+a_{j+1}$
and so acting by any $R_j$ leaves the sum of the sequence invariant.
\end{rem}
\begin{defn}
Define $R:\mathbb R\to\mathbb R$ by  $R=\cdots\circ R_3\circ R_2\circ R_1\circ\cdots$.
\end{defn}
\begin{rems}
\begin{enumerate}
  \item For any sequence $a$, the resulting sequence $R(a)$ has terms which do not exceed 1 and so $R$ is idempotent: $R^2(a)=R(a)$.
  
  \item Although $R$ is the composition of infinitely many functions, all but finitely many of them act as the identity. Since for any given $a$, $a_i\to0$ as $i\to-\infty$ then certainly $a_i\le1$ for $i<m$ (say). Also, since the sum of the sequence is finite, it follows that for some $n\ge m$,
  \begin{equation}
  R(a)=R_n\circ\cdots\circ R_{m+1}\circ R_{m}(a).
  \end{equation}
  
\end{enumerate}
\end{rems}  

\begin{exa}
Consider sequence $a=(\dots,0,\tfrac12,\tfrac52,-\tfrac16,-\tfrac13,0,\dots)$ where the first non-zero term is $a_1$. The smallest $i$ for which $a_i>1$ is $2$. Then 
\begin{align*}
R(a)&=\dots\circ R_4\circ R_3\circ R_2\circ R_1\circ\cdots(\dots,0,\tfrac12,\tfrac52,-\tfrac16,-\tfrac13,0,\dots)\\
    &=\dots\circ R_4\circ R_3\circ R_2(\dots,0,\tfrac12,\tfrac52,-\tfrac16,-\tfrac13,0,\dots)\\
    &=\dots\circ R_4\circ R_3(\dots,0,\tfrac12,1,\tfrac43,-\tfrac13,0,\dots)\\
    &=\dots\circ R_4(\dots,0,\tfrac12,1,1,0,0,\dots)\\
    &=(\dots,0,\tfrac12,1,1,0,0,\dots).
\end{align*}
\end{exa}
\begin{lem}\label{lem:R}
For any summable sequence $a$, 
\begin{equation}\label{eq:R}
R(a)_i=\min\left(1,\sum_{j\le i}a_j-\sum_{j<i}R(a)_j\right).  
\end{equation}
\end{lem}
\begin{proof}
Let $m$ be the smallest index such that $a_m>1$ and for each $j$, where $j\geq m$, define $a^{(j)}=R_j\circ R_{j-1}\circ R_{j-2}\circ\cdots\circ R_m(a)$. 

We first observe that for any $i$, $a^{(i)}_k=a_k$ for all $k>i+1$ and $a^{(i)}_k=a^{(k)}_k$ for all $k<i$. It follows that for any $i$, $R(a)_j=a^{(j)}_j$. Also we have 
\begin{equation}\label{eq:alt ind 1}
a_i^{(i)}=\min(1,a^{(i-1)}_i)\text{ and }a_{i+1}^{(i)}=a_{i+1}+a^{(i-1)}_{i}-a^{(i)}_{i}.    
\end{equation}  
It follows by a straightforward inductive argument starting at $i=m$ that 
\begin{equation}
a_i^{(i-1)}=\sum_{j\le i}a_j-\sum_{j<i}a^{(j)}_j,
\end{equation}  
and then \eqref{eq:R} follows from \eqref{eq:alt ind 1}.
\end{proof}

\begin{thm}
A real summable sequence $(U^t_i)$ has udKdV time update $(U^{t+1}_i)$ where
\begin{equation}
  U^{t+1}_i+U^t_i=R(U^t_{i-1}+U^t_i).
\end{equation}
\end{thm}
\begin{proof}
Let $a$ be the sequence with $a_i=U^t_{i-1}+U^t_i$. Then, by Lemma~\ref{lem:R},
\begin{align*} 
R(a)_i-U^t_i&=\min\left(1,\sum_{j<i}U^t_j+\sum_{j\le i}U^t_j-\sum_{j<i}R(a)_j\right)-U^t_i\\
          &=\min\left(1-U^t_i,\sum_{j<i}U^t_j-\sum_{j<i}(R(a)_j-U^t_j)\right),
\end{align*} 
which is the time update rule \eqref{eq:update rule} with $R(a)_i-U^t_i=U^{t+1}_i$.
\end{proof}
\begin{rem}
There is an entirely analogous alternative version of the downdate rule in which excess cell capacity is moved left rather than right:
\begin{equation}
  U^{t-1}_i+U^t_i=L(U^t_i+U^t_{i+1}),
\end{equation}
where $L=\cdots\circ L_1\circ L_2\circ L_3\circ\cdots$ and
\begin{equation}
L_j(a_j)=\min(1,a_j),\ \ L_j(a_{j-1})=a_{j-1}+a_j-\min(1,a_{j}),\ \ L_j(a_i)=a_i\ (i\ne j,j-1).
\end{equation}
\end{rem}

\begin{exas}
This example will illustrate how to use this version of the downdate rule. We use an over bar to indicate negative numbers and values of $U^t_i$ outside the interval shown are 0.
\[
\setlength{\arraycolsep}{0pt}
\renewcommand{\arraystretch}{1.5}
  \begin{array}{rCCCCCCCCCCCCCCCCCCCCCCCCCCCCC}
    U^t_i\ :&&0&0&0&1&\frac12&0&\frac12&1&1&\overline{\frac12}&0&0\\
    U^t_i+U^t_{i+1}\ :&&0&0&1&\frac32&\frac12&\frac12&\frac32&2&\frac12&\overline{\frac12}&0&0\\
    L(U^t_i+U^t_{i+1})\ :&&0&1&1&1&1&1&1&1&\frac12&\overline{\frac12}&0&0\\
    L(U^t_i+U^t_{i+1})-U^t_i=U^{t-1}_i\ :&&0&1&1&0&\frac12&1&\frac12&0&\overline {\frac12}&0&0&0\\
   \end{array}
\]
A second example illustrates the update rule with real values. 
\[
\setlength{\arraycolsep}{1pt}
\renewcommand{\arraystretch}{1.5}
  \begin{array}{rCCCCCCCCCCCCCCCCCCCCCCCCCCCCC}
    U^t_i\ :&&0&\pi&0&0&0&0&0&0&0\\
    U^t_{i-1}+U^t_{i}\ :&&0&\pi&\pi&0&0&0&0&0&0\\
    R(U^t_{i-1}+U^t_{i})\ :&&0&1&1&1&1&1&1&\alpha&0\\
    R(U^t_{i-1}+U^t_{i})-U^t_i=U^{t+1}_i\ :&&0&\beta&1&1&1&1&1&\alpha&0\\
   \end{array}
\]
where $\alpha=2\pi-6\in(0,1)$ and $\beta=1-\pi<0$. Note that the total masses in the two examples, $\frac72$ and $\pi=5+\alpha+\beta$ respectively, are preserved under the evolution.
\end{exas}

Using the alternative formulation of the update/downdate rule the following properties are apparent:
\begin{cor}\label{cor:cons}\par\noindent
\begin{enumerate}
  \item Total mass $\sum_{i\in\mathbb Z}U^t_i$ is conserved \uplp cf. \eqref{eq:cons u}\uprp.
  \item For any $t$, $V^t$
   defined in equation \eqref{eq:V^t} satisfies: $V^t\le1\iff U^{t+1}_i=U^t_{i-1}$ for all $i$. If $V^t\le1$ then it is $t$-invariant.
\end{enumerate}
\end{cor}
\begin{proof}
(i) This follows since the action of $R$ preserves the sum of any sequence. 

(ii) A sequence $a$ is invariant under the action of $R$ if and only if $a_i\le1$ for all $i$ and $U^{t+1}=R(U^t_{i-1}+U^t_i)-U^t_i=U^t_{i-1}\iff R$ acts as the identity.
\end{proof}

\section{Two more conserved quantities}\label{sec:conserved}
We have already seen that when $U^t_i$ satisfies udKdV then the sum $\sum_iU^t_i$ is conserved. In this section we establish two more constants of the motion. These conserved quantities have essential applications to the study of the solution of the ultradiscrete linear problem \eqref{eq:lin1}--\eqref{eq:lin4}. Some proofs given in this section make use of the following general result. 
\begin{lem}\label{lem:seq}
Let $(a_n)$, $(b_n)$, $(n\in\mathbb Z)$ be bounded real sequences. If there exists \uplp at least one\uprp{} $m$ such that $a_m>a_{m-1}$, and for every such $m$, $b_m\ge a_m$ then $\sup_i b_i\ge \sup_i a_i$.

By reversing the order of these sequences this criterion is equivalently restated: if there exists \uplp at least one\uprp{} $m$ such that $a_m>a_{m+1}$, and for every such $m$, $b_m\ge a_m$ then $\sup_i b_i\ge \sup_i a_i$.
\end{lem}
\begin{proof}
Let $I$ be the (necessarily non-empty) subset of $\mathbb Z$ at which the sequence $(a_n)$ increases, that is $I=\{m\in\mathbb Z:a_m>a_{m-1}\}$. It is obvious that $\sup_{i\in\mathbb Z} a_i=\sup_{i\in I} a_i$ and if $b_m\ge a_m$ for all $m\in I$ then $\sup_{i\in I} b_i\ge\sup_{i\in I} a_i$. Hence $\sup_{i\in I} b_i\ge\sup_{i\in\mathbb Z} a_i$.
\end{proof}

Now consider two particular sequences defined by
\begin{equation}\label{eq:X}
  X^t_i :=\sum_{j\ge i}(U^{t+1}_j-U^{t-1}_j)=\sum_{j<i}(U^{t-1}_j-U^{t+1}_j),
\end{equation}
and
\begin{equation}\label{eq:Y}
Y^t_i :=\sum_{j>i}U^{t}_j-\sum_{j\ge i}U^{t-1}_j=\sum_{j<i}U^{t-1}_j-\sum_{j\le i}U^{t}_j,
\end{equation}
where the alternative expressions are obtained using the conservation of total mass \eqref{eq:cons U}. Since all time updates and downdates of $U^t_j$ equal zero for $|j|$ large enough, both $X^t_i$ and $Y^t_i$ are zero for $|i|$ sufficiently large. These sequences therefore attain maximum values. Further, they will be seen to be the ultradiscrete conserved densities of two more conserved quantities: $\max_i X^t_i$ is the \emph{maximum soliton mass} and $1+\max_i Y^t_i$ is the \emph{maximum soliton speed}. A discrete conserved density is a sequence the \emph{sum} of which is time independent; an ultradiscrete conserved density is a sequence the \emph{maximum} of which is time independent.

We will now show some properties that will allow us to apply Lemma~\ref{lem:seq} to establish the connection between the maximum values of these two sequences. First we record a few basic formulae. The notation $\Delta$ for the forward difference operator $\Delta A_i=A_{i+1}-A_i$ is used:
\begin{align}
  \Delta X^t_i&=U^{t-1}_{i}-U^{t+1}_{i}\label{eq:basic 1}\\
  \Delta Y^t_i&=U^{t-1}_{i}-U^{t}_{i+1}\label{eq:basic 2}\\
  X^t_i&=U^{t}_{i}+U^{t+1}_{i}+Y^t_{i}+Y^{t+1}_{i}=U^{t-1}_{i-1}+U^{t}_{i-1}+Y^t_{i-1}+Y^{t+1}_{i-1}.\label{eq:basic 3}
\end{align}
The downdate rule \eqref{eq:downdate rule} may be written in terms of $Y^t_i$ as 
\begin{equation}
  \min(1-U^t_i-U_i^{t-1},Y^t_i)=0,
\end{equation}
and so it is clear that for all $i,t$ 
\begin{equation}\label{eq:basic}
  Y^t_i\ge0\quad\text{and}\quad U^t_i+U_i^{t-1}\le1 ,
\end{equation}
and
\begin{equation}
\label{eq:basic 4}
  Y^t_i>0\implies U^t_i+U^{t-1}_i=1,\quad U^t_i+U^{t-1}_i<1\implies Y^t_i=0.
\end{equation}

\begin{lem}\label{lem:V^t}
The following three statements are equivalent for any $t$:
\begin{enumerate}
  \item $V^t\le1$,
  \item for all $i$, $U^{t+1}_i=U^t_{i-1}$,
  \item for all $i$, $Y^t_i=0$.
\end{enumerate}
Moreover, these three conditions are $t$-invariant.
\end{lem}
\begin{proof}
The equivalence and $t$-invariance of the first two statements were established in Corollary~\ref{cor:cons} (ii). 
Suppose that $Y_i^t=0$ for all $i$, at some $t$. Then relation \eqref{eq:basic 2} gives $U^{t-1}_{i}=U^{t}_{i+1}$,  which is equivalent to (ii) since this condition is $t$-invariant by Corollary~\ref{cor:cons}. Conversely,  if condition (ii) holds then for all $i$ and $t$ we have $U^{t-1}_{i}=U^{t}_{i+1}$ and $Y_i^t=0$, by \eqref{eq:basic 2}, for all $i$ at every instant $t$.
\end{proof}
This Lemma implies:
\begin{cor}\label{cor:newcor}\par\noindent
\begin{enumerate}
  \item If the conditions in Lemma~\ref{lem:V^t} are met then: for all $i$, $X^t_i\le V^t$, with equality for some $i$.
  \item If $V^t>1$ then there exists some $i$ for which $X_i^t>1$.
\end{enumerate}
\end{cor}
\begin{proof}
(i) Suppose that $U^{t+1}_i=U^t_{i-1}$ for all $i$. Then \eqref{eq:X} gives $X^t_i=U^t_i+U^t_{i-1}$ which is of course less than or equal to $\max_i(U^{t}_i+U^t_{i+1})$. Hence, $X^t_i\le V^t$ \eqref{eq:V^t}, and the equality must be attained for some $i$. Notice that these relations are $t$-invariant.

(ii) By Lemma \ref{lem:V^t}, if $V^t>1$ then $Y_i^t>0$ for some $i$. Then by \eqref{eq:basic 4}, $U^t_i+U^{t-1}_i=1$ and by \eqref{eq:basic 3}, $X^t_{i+1}=1+Y^t_i+Y^{t+1}_{i}$. Since $Y_i^t>0$ and $Y^{t+1}_i\ge0$, $X^t_{i+1}>1$. 
\end{proof}
{The second inequality} in \eqref{eq:basic} can be slightly sharpened. In the case $V^t=\max_i(U^{t}_i+U^t_{i+1})\le1$ we have $U^{t-1}_i=U^t_{i+1}$ and hence for any $U^t_i$, 
\begin{equation}\label{eq:basic gen}
  U^t_i+U_i^{t-1}\le\min(1,V^t). 
\end{equation}
While investigating the relationship between the maxima of $X^t_i$ and $Y^t_i$ we will consider two cases, $V^t\le1$ and $V^t>1$. In the case $V^t\le1$, it is clear from {Lemma~\ref{lem:V^t} and Corollary~\ref{cor:newcor} } that
\begin{equation}\label{eq:<=1 max}
\max_i Y^t_i=0\text{ and \/}\max_iX_i^t=V^t\text{ for all $t$}.
\end{equation}
We will see below that this is a special case of a result that holds in all cases: $\max_iX_i^t$ and $1+\max_i Y^t_i$ are independent of $t$ and give the mass and speed of the heaviest soliton in $U^t_i$.

In the next two propositions we consider the case in which $V^t>1$. Notice that this implies that $V^t>1$ for all $t$ (since any value $V^t\leq1$ is $t$-invariant). So, by {Corollary~\ref{cor:newcor}  and Lemma~\ref{lem:V^t}}, for any $t$ there exists $i$, such that $X^t_i>1$ and $Y^t_i>0$.

\begin{prop}\label{prop:1}
Let $V^t>1$. 
\begin{enumerate}
  \item For any $t$, there exists $m$ such that $Y^t_m>Y^t_{m-1}$ and for any such $m$, $X^t_{m+1}\ge 1+Y_m^t$ and so 
\begin{equation}
  \max_i X^t_i\ge 1+\max_i Y^t_i.
\end{equation}
\item For any $t$, there exists $m$ such that $X^t_{m+1}>X^t_{m}$ and for any such $m$, $1+Y_m^t\ge X^t_{m+1}$ and so 
\begin{equation}
   1+\max_i Y^t_i\ge\max_i X^t_i.
\end{equation}
\end{enumerate}
Hence 
\begin{equation}
  \max_i X^t_i=1+\max_i Y^t_i.
\end{equation}
\end{prop}
\begin{proof}
In both cases, the existence of $m$ is clear (since $V^t>1$ and $X_i^t, Y_i^t=0$ for $|i|$ sufficiently large) and the conclusions are obtained using the principle established in Lemma~\ref{lem:seq}. In (i) we take $a_n=1+Y^t_n$ and $b_n=X^t_{n+1}$, and in (ii) $a_n=X^t_{n+1}$ and $b_n=1+Y^t_n$. We now establish the required inequalities.

\begin{enumerate}
  \item Let $Y^t_m>Y^t_{m-1}$ (i.e. $a_m>a_{m-1}$). From \eqref{eq:basic}, all $Y^t_i\ge0$ and it follows that $Y^t_m>0$. Using \eqref{eq:basic 3} and \eqref{eq:basic 4}
\begin{equation}
X^t_{m+1}=U^{t-1}_m+U^{t}_m+Y^t_m+Y^{t+1}_m=1+Y^t_m+Y^{t+1}_m\ge1+Y^t_m,
\end{equation}
(i.e. $b_m\ge a_m$) and the result follows.
  \item Similarly, let $X^t_{m+1}>X^t_{m}$. Then from \eqref{eq:basic 1}, $U^{t+1}_{m}<U^{t-1}_{m}$ and since from \eqref{eq:basic} $U^t_{m}+U^{t-1}_{m}\le1$, we obtain $U^{t+1}_{m}+U^{t}_{m}<1$ and so from \eqref{eq:basic 4} $Y^{t+1}_{m}=0$. Therefore, using \eqref{eq:basic 3},
\begin{equation*}
X^t_{m+1}=U^{t-1}_m+U^{t}_m+Y^t_m + Y^{t+1}_m\le 1+Y^t_m.\qedhere
\end{equation*}
\end{enumerate}
\end{proof}

\begin{prop}\label{prop:2}
Let $V^t>1$. 
\begin{enumerate}
  \item For any $t$, there exists $m$ such that $Y^{t+1}_m>Y^{t+1}_{m+1}$ and for any such $m$, $X^t_{m}\ge 1+Y_m^{t+1}$ and so 
\begin{equation}
  \max_i X^t_i\ge 1+\max_i Y^{t+1}_i.
\end{equation}
\item For any $t$, there exists $m$ such that $X^t_{m}>X^t_{m+1}$ and for any such $m$, $1+Y_m^{t+1}\ge X^t_{m}$ and so 
\begin{equation}
   1+\max_i Y^{t+1}_i\ge\max_i X^t_i.
\end{equation}
\end{enumerate}
Hence 
\begin{equation}
  \max_i X^t_i=1+\max_i Y^{t+1}_i.
\end{equation}
\end{prop}
\begin{proof}
Again the existence of $m$ is clear and as in the proof of Proposition~\ref{prop:1}, we apply Lemma~\ref{lem:seq}.  In proving (i) we take $a_n=1+Y^{t+1}_n$ and $b_n=X^t_{n}$, and in (ii), $a_n=X^t_{n}$ and $b_n=1+Y^{t+1}_n$.
\begin{enumerate}
  \item Let $Y^{t+1}_m>Y^{t+1}_{m+1}$. Hence $Y^{t+1}_m>0$ and so using \eqref{eq:basic 3} and \eqref{eq:basic 4}
\begin{equation}
X^t_{m}=U^{t+1}_m+U^{t}_m + Y^t_m+Y^{t+1}_m =1+Y^t_m+Y^{t+1}_m\ge1+Y^{t+1}_m.
\end{equation}

\item Let $X^t_{m}>X^t_{m+1}$. Then from \eqref{eq:basic 1}, $U^{t-1}_{m}<U^{t+1}_{m}$ and since $U^t_{m}+U^{t+1}_{m}\le1$, we obtain $U^{t-1}_{m}+U^{t}_{m}<1$ and so from \eqref{eq:basic 4} $Y^{t}_{m}=0$. Therefore, using \eqref{eq:basic 3},
\begin{equation*}
X^t_{m}=U^{t+1}_m+U^{t}_m+Y^t_m+Y^{t+1}_m\le 1+Y^{t+1}_m.\qedhere
\end{equation*}

\end{enumerate}
\end{proof}
The results of this section so far are summarised in a theorem.
\begin{thm}\label{thm:cons X Y}
Let $U^t_i$ satisfy the udKdV equation \eqref{eq:update rule}. Then $\max_i X^t_i$ and $\max_i Y^t_i$, with $X^t_i$ and $Y^t_i$ as defined in \eqref{eq:X} and \eqref{eq:Y}, are conserved under time evolution. 

Moreover, if $V^t>0$,
\begin{equation}\label{eq:max}
  \omax=\max_i X^t_i\quad\text{and}\quad1+\max_iY^t_i=c_{\text{max}}=\max(1,\omax),
\end{equation}
are the mass and speed of the solitons of maximal mass contained in the state $U_i^t$.
\end{thm}

\begin{rem}\label{rem:omega max}
We define $\kmax:=\min(1,\omax)$. Then, in the case $0<V^t\le1$, from \eqref{eq:<=1 max}, $\omax=\kmax=V^t\leq1$ and $c_\text{max}=\omax/\kmax=1$. Such a state $U_i^t$ consists of solitons that move in tandem with the background, at speed 1. The maximal mass $\omax$ of the solitons in such a state can therefore not be ascertained by mere (visual) inspection of $U_i^t$ if the solitons are engulfed in the background. 

For $V^t>1$ however \uplp from Corollary \ref{cor:newcor}\uprp\, we have $\omax>1$ and $\kmax=1$ and $\omax$ (being independent of $t$) may be calculated asymptotically as $t\to\pm\infty$. A more rigorous analysis of this case, in terms of the ultradiscrete squared eigenfunction for such a state will be presented later on, but there is also the following  simple heuristic argument. As discussed in Section \ref{sec:sol prop}, arbitrary initial conditions $U^0_i$ evolve into a finite train of solitons each characterised by a mass $\omega>0$ and phase $\phi$. For $V^t>1$, the solitons of mass $\omega>1$ move at speed $c=\omega$ and hence become well separated as $t\to\pm\infty$ whereas solitons of mass $\omega\le1$ move at speed 1 along with the background. As we shall see, the ultradiscrete spectrum  is not simple and an arbitrary number of solitons may have the same mass $\omega$ and speed. However such sets of solitons have a minimum separation, also equal to $\omega$. For large enough $t$ all solitons are therefore well separated and it is straightforward to use \eqref{eq:exp sol} to show that any local maximum of $X^t_i$ is equal to the mass of the soliton which includes site $i$, and is sub-maximal otherwise. The interpretation of $\omax$ as maximal soliton mass follows.
\end{rem}

\begin{rem}\label{newposremark4}
Using the quantity $\kmax=\min(1,\omax)$ inequality \eqref{eq:basic gen} can be reformulated, for general $U_i^t$, as
\begin{equation}\label{eq:basic kmax}
U^t_i+U_i^{t-1}\le\kmax,
\end{equation} 
for any $i,t$.
\end{rem}

\begin{rem}\label{newposremark5}
If $U^t_i$ is nontrivial but $V^t=0$ at some instant $t$, it is 0 for all $t$ and we have just background propagating unchanged at speed 1. There are no solitons: since $V^t=0$, from Lemma \ref{lem:V^t} and Corollary \ref{cor:newcor} we have $\omax=0$. Moreover, $\kmax=\min(1,\omax)=0$. Note that the converse is also true. If $\omax=0$, then from Theorem \ref{thm:cons X Y} we have $Y_i^t=0$ for all $i$ and thus, by Corollary \ref{cor:newcor}, $U_{i+1}^{t+1} = U_i^t$, for all $i$ and $t$. From remark \ref{newposremark4} we then have $V^t=0$ since $U_i^t+U_{i+1}^t=U_i^t+U_i^{t-1}\le 0$ for all $i$.
\end{rem}

\begin{prop}\label{prop:X ineq}
For any $U^t_i$ we have
\begin{equation}\label{eq: X ineq}
  \Delta X_i^t>0\implies U^t_i+U^{t+1}_i<\kmax,\quad
  \Delta X_i^t<0\implies U^t_i+U^{t-1}_i<\kmax.
\end{equation}
\end{prop}
\begin{proof}
From \eqref{eq:basic 1}, $\Delta X^t_m=U^{t-1}_i-U^{t+1}_i$. Then, from \eqref{eq:basic kmax},
\begin{equation}
	\Delta X^t_m\le\kmax-U^{t}_i-U^{t+1}_i\text{\ \ or \ \ }		\Delta X^t_m\ge U^{t}_i+U^{t-1}_i-\kmax.	
\end{equation}
The results follow immediately.
\end{proof}
\begin{exas}\label{exa:omax}
Consider a state $U^t_i$. The first value shown has index $i=1$ and the state is 0 outside of the interval shown. We use an over bar to indicate a negative number: for $x>0$, $\overline{x}$ denotes $-x)$.
\[
\setlength{\arraycolsep}{0pt}
\renewcommand{\arraystretch}{1.5}
  \begin{array}{rCCCCCCCCCCCCCCCCCCCCCCCCCCCCCCCCCCCCCC}
U^t_i\ :&&0&0&0&0&1&1&0&1&1&0&0&0&0&1&1&1&0&0&0&0\\
U^{t-1}_i\ :&&0&1&1&1&0&0&1&0&0&0&1&1&1&0&0&0&0&0&0&0\\
U^{t+1}_i\ :&&0&0&0&0&0&0&1&0&0&1&1&1&0&0&0&0&1&1&1&0\\
U^{t+1}_i-U^{t-1}_i\ :&&0&\overline1&\overline1&\overline1&0&0&0&0&0&1&0&0&\overline1&0&0&0&1&1&1&0\\
\sum_{j\ge i}(U^{t+1}_j-U^{t-1}_j)=X^t_i\ :&&0&0&1&2&3&3&3&3&3&3&2&2&2&3&3&3&3&2&1&0
\end{array}
\]
Hence $\omax=3$, which is attained throughout two different regions (in $i$). This indicates the presence of at least two solitons of mass 3. In principle, there may be several maximal mass solitons within each maximal region, as another example illustrates:
\[
\setlength{\arraycolsep}{0pt}
\renewcommand{\arraystretch}{1.5}
  \begin{array}{rCCCCCCCCCCCCCCCCCCCCCCCCCCCCCCCCCCCCCC}
U^t_i\ :&&0&0&0&0&1&1&0&0&1&1&0&0&0&0&0\\
U^{t-1}_i\ :&&0&0&1&1&0&0&1&1&0&0&0&0&0&0&0\\
U^{t+1}_i\ :&&0&0&0&0&0&0&1&1&0&0&1&1&0&0&0\\
U^{t+1}_i-U^{t-1}_i\ :&&0&0&\overline1&\overline1&0&0&0&0&0&0&1&1&0&0&0\\
X^t_i\ :&&0&0&0&1&2&2&2&2&2&2&2&1&0&0&0
\end{array}
\]
giving $\omax=2$, attained throughout a single region although $U_i^t$ contains two mass 2 solitons.
Finally, we give an example with non-binary values: 
\[
\setlength{\arraycolsep}{0pt}
\renewcommand{\arraystretch}{1.5}
  \begin{array}{rCCCCCCCCCCCCCCCCCCCCCCCCCCCCCCCCCCCCCC}
U^t_i\ :&&0&0&\overline{\frac{1}{2}}&\frac{1}{3}&\frac{5}{3}&0&0&1&0&0&0\\
U^{t-1}_i\ :&&0&0&\frac{3}{2}&\frac{2}{3}&\overline{\frac{2}{3}}&0&1&0&0&0&0\\
U^{t+1}_i\ :&&0&0&0&\overline{\frac{1}{2}}&\overline{\frac{2}{3}}&1&1&0&1&\frac{2}{3}&0\\
U^{t+1}_i-U^{t-1}_i\ :&&0&0&\overline{\frac{3}{2}}&\overline{\frac{7}{6}}&0&1&0&0&1&\frac{2}{3}&0\\
X^t_i\ :&&0&0&0&\frac{3}{2}&\frac{8}{3}&\frac{8}{3}&\frac{5}{3}&\frac{5}{3}&\frac{5}{3}&\frac{2}{3}&0
\end{array}
\]
giving $\omax=\frac83$, indicating the presence of a mass $\frac83$ soliton, something which is not immediately visible on the state $U_i^t$ (or on its up- or downdates).
\end{exas}

The dynamics in $t$ of the regions where $X_i^t$ takes its maximal values seems interesting but is probably too complicated to allow for a complete description. We  can however show the following Proposition and Corollary.

\begin{prop}\label{lem:m-1} For any $U^t_i$ we have that
\begin{enumerate}
  \item If $\Delta X_i^t<0$ then $X_i^{t-1}<\omax$, and if $\,\Delta X_i^{t-1}>0$ then $X_{i+1}^t<\omax$.
  \item If $m$ is the left-most index of a region in $i$ where $X_i^t$ is maximal, i.e. if $X_m^t=\omax$ and $\Delta X_{m-1}^t>0$, then $X_{m-1}^{t-1}=\omax$.
  \item If $m$ is the right-most index of a region in $i$ where $X_i^t$ is maximal, i.e. if $X_m^t=\omax$ and $\Delta X_{m}^t<0$, then $X_{m+1}^{t+1}=\omax$.
\end{enumerate}  
\end{prop}
\begin{proof}
\begin{enumerate}
  \item Suppose first that $\Delta X_i^t<0$. From Proposition \ref{prop:X ineq} we then have that 
  \begin{equation}
   U_i^t + U_i^{t-1} < \kappa_{\max}=\min(1,\omax)\leq 1,
  \end{equation}
  and thus from \eqref{eq:basic 4} that $Y_i^t=0$. Expressing $X_i^{t-1}$ by means of \eqref{eq:basic 3}, we then find
  \begin{equation}
  X_i^{t-1} = U_i^t + U_i^{t-1} + Y_i^{t-1} < \min(1,\omax) + \big(\max(1,\omax) -1\big) = \omax,
  \end{equation}
  because of Theorem \ref{thm:cons X Y}. 
  Next, suppose that $\Delta X_i^{t-1}>0$. As above, we find that $U_i^t + U_i^{t-1}<\min(1,\omax)$ and $Y_i^t=0$, and by expressing $X_{i+1}^{t}$ by means of \eqref{eq:basic 3} that:
  \begin{equation}
  X_{i+1}^{t} = U_i^t + U_i^{t-1} + Y_i^{t+1} < \min(1,\omax) + \big(\max(1,\omax) -1\big) = \omax.
  \end{equation}
  \item Using \eqref{eq:basic 3} it is easily checked that the difference of $X_m^t$ and  $X_{m-1}^{t-1}$ can be expressed as
\begin{equation}
X_m^t -  X_{m-1}^{t-1} = Y_{m-1}^{t+1} - Y_{m-1}^{t-1}.
\end{equation}
Moreover, since $X_m^t=\omax=\max_i X_i^{t'}$ for all times $t'$ (Theorem \ref{thm:cons X Y}), we have that this difference is always non-negative and as a consequence that
\begin{equation}\label{Ydiff}
Y_{m-1}^{t+1} \geq Y_{m-1}^{t-1}.
\end{equation}
As $\Delta X_{m-1}^t>0$ we have from Proposition \ref{prop:X ineq} that $U_{m-1}^t+ U_{m-1}^{t+1}$ is less than $\kappa_{\max}$, i.e. less than 1, and \eqref{eq:basic 4} then tells us that $Y_{m-1}^{t+1}=0$.
Because of \eqref{Ydiff} we then find that $Y_{m-1}^{t-1}=0$ as well, since we know from \eqref{eq:basic} that it cannot be negative. We therefore have $Y_{m-1}^{t+1} = Y_{m-1}^{t-1} =0$ and $X_{m-1}^{t-1} = X_m^t = \omax$.
  \item Similarly, we have $X_{m+1}^{t+1} -  X_{m}^{t} = Y_{m}^{t+2} - Y_{m}^{t}$ and from $X_{m+1}^{t+1} \leq  X_{m}^{t}$ we obtain $0\leq Y_{m}^{t+2} \leq Y_{m}^{t}$. Then, as $\Delta X_m^t<0$ we have from Proposition \ref{prop:X ineq} that  $U_{m}^t+ U_{m}^{t-1} <1$ and thus $Y_m^t=0$ by \eqref{eq:basic 4}. We therefore also have $Y_{m}^{t+2}=0$ which shows that $X_{m+1}^{t+1} =  X_{m}^{t} =\omax$.
\end{enumerate}
\end{proof}

The above results show that although the lengths of the regions where $X_i^t$ is maximal might change over time, and although a region in $X_i^{t-1}$ might have some overlap with a region in $X_i^t$, this overlap must stop before the right-most index in that region for $X_i^t$ (Proposition \ref{lem:m-1} (i)). The same clause in that proposition also tells us that a region where $X_i^t$ is maximal cannot extend into such a region in $X_i^{t-1}$. Furthermore, clauses (ii) and (iii) in Proposition \ref{lem:m-1} tell us that such regions, at different time steps, cannot be too far apart either: before each region (counted in $i$) where $X_i^t$ is maximal there is a similar region in $X_i^{t-1}$, not further removed than one index (Proposition \ref{lem:m-1} (ii)) and to the right of such a region in $X_i^t$ there is a similar region in $X_i^{t+1}$ not further away than one index (Proposition \ref{lem:m-1} (iii)).
\begin{cor}
The number of distinct regions in $i$ where $X_i^t$ takes its maximal value is the same for all $t$.
\end{cor}

Finally in this section we state a technical lemma that is needed later.  A proof is given in Appendix~\ref{ap:proof of 9}.

\begin{lem}\label{lem:m}
\begin{enumerate}
  \item If $X^t_m>1$ is locally maximal at $m$, at time $t$, then 
\begin{equation}\label{eq:lem 9 1}
(U^t_{m-1}+U^t_{m}\ge1\text{ or \/}U^{t-1}_m+U^t_{m}=1)\text{ and \/} (U^t_{m-1}+U^t_{m}\ge1\text{ or \/}U^t_{m-1}+U^{t+1}_{m-1}=1).
\end{equation} 
  \item If $X_m^t$ is a global maximum, i.e. if $X_m^t=\omax$, then
\begin{equation}\label{eq:lem 9 2}
U^t_{m-1}+U^{t-1}_{m-1}=\kmax.
\end{equation} 
\end{enumerate}
\end{lem}

\section{On Darboux transformations for dKdV and udKdV}\label{sec:darboux}

Let $u^t_i$ be a solution of the dKdV equation \eqref{eq:dkdv 1} and let $\theta^t_i$ be a non-zero solution of linear system \eqref{eq:dlin1}, \eqref{eq:dlin2} for any choice of parameter $\lambda$. Then \cite{MR3150627} 
\begin{equation}
\label{eq:darboux phi}
\tilde\phi^t_i=\frac{\phi_{i-1}^t\theta^t_i-\phi_{i}^t\theta_{i-1}^t}{\theta^t_i},
\end{equation}
satisfies \eqref{eq:dlin1}, \eqref{eq:dlin2} with 
\begin{equation}\label{eq:DT tau}
\tilde \tau^t_i=\tau^t_i\theta_i^t 
\end{equation}
and thus, using \eqref{eq:u tau},
\begin{equation}\label{eq:DT u}
\tilde u^t_i=u^t_i\frac{\theta_i^{t+1}\theta_{i+1}^{t}}{\theta_i^t\theta_{i+1}^{t+1}},
\end{equation}
is a solution of the dKdV equation \eqref{eq:dkdv 1}.

Because of the presence of a negative sign, there is no obvious ultradiscrete counterpart of \eqref{eq:darboux phi}. However, the transformation of the tau function \eqref{eq:DT tau} and potential \eqref{eq:DT u} have ultradiscrete limits
\begin{equation}\label{eq:ud darboux}
\tilde T^t_i=T^t_i+\Theta^t_{i}\quad\text{and}\quad
\tilde U^t_i=U^t_i+\Theta^t_{i+1}+\Theta^{t+1}_i-\Theta^t_{i}-\Theta^{t+1}_{i+1}.
\end{equation}
In Sections \ref{sec:line sol} and \ref{sec:bound states} we will construct solutions to the linear problem \eqref{eq:lin1}--\eqref{eq:lin4} for which \eqref{eq:ud darboux} indeed acts as an ultradiscrete  Darboux transformation, i.e.: it maps a solution $U^t_i$ of udKdV \eqref{eq:update rule} with finite support to another solution of udKdV, also with finite support (and similarly for the associated $T$-functions). 

We now derive some properties of the discrete Darboux transformation and the corresponding properties they imply for the ultradiscrete Darboux transformation.
\begin{lem}
Let $\theta^t_i$, any non-zero solution of \eqref{eq:dlin1}, \eqref{eq:dlin2}, determine the Darboux transformation \eqref{eq:DT u} from $u^t_i$ to $\tilde u^t_i$. Then
\begin{equation}
\label{eq:DT id}
\frac1{\tilde u^t_i}+\delta\tilde u^t_{i-1}=
\frac{\theta^t_i\strut^2}{\theta^t_{i+1}\theta^t_{i-1}}
\left(\frac1{u^t_i}+\delta u^t_{i-1}\right).
\end{equation}
The transformed linear problem has particular solution 
\begin{equation}\label{eq:1/theta}
\tilde\theta^t_i:=\frac{1}{\rho^t_i}=\frac{(1-\lambda^2)^{-i}(1-\lambda^2/\beta^2)^{t}}{\theta^t_i},   
\end{equation}
where $\rho^t_i$ is the adjoint eigenfunction related to $\theta^t_i$ \uplp cf.\ \eqref{eq:adjoint}\uprp.
The Darboux transformation defined by $\tilde\theta^t_i$ maps $\tilde u^t_i$ back to $u^t_i$.
\end{lem}
\begin{proof}
From \eqref{eq:dlin1} and \eqref{eq:dlin2} the $\theta^t_i$
satisfy
\begin{gather}
  \label{eq:dlin1 th}
  \theta_{i-1}^{t}-\frac2{1+\delta}\left(\frac1{u_{i}^{t}}+\delta u_{i-1}^{t}\right)\theta_{i}^{t}+(1-\lambda^2)\theta_{i+1}^{t}=0\\
  \label{eq:dlin2 th}
  \beta\theta_{i}^{t+1}=\theta_{i-1}^{t}+(\beta-1)u_{i-1}^{t}\theta_{i}^{t}=-(1-\lambda^2)\theta_{i+1}^{t}-\frac{(\beta-1)}{\delta u_{i}^{t}}\theta_{i}^{t},
\end{gather}
for some parameter $\lambda$, with $\beta = \dfrac{1-\delta}{1+\delta}$.
Using \eqref{eq:dlin2 th}, we get
\begin{align*}
\frac1{\tilde u^t_i}+\delta\tilde u^t_{i-1}&=\frac{\theta^t_i}{\theta^t_{i+1}\theta^t_{i-1}\theta^{t+1}_i}
\left(\frac1{u^t_i}\theta^{t+1}_{i+1}\theta^t_{i-1}+\delta u^t_{i-1}\theta^t_{i+1}\theta^{t+1}_{i-1}\right)\\
&=\frac{\theta^t_i\strut^2}{\beta\theta^t_{i+1}\theta^t_{i-1}\theta^{t+1}_i}
\left(
\frac1{u^t_i}\theta^t_{i-1}-
\delta u^t_{i-1}\theta^t_{i+1}(1-\lambda^2)\right)\\
&=\frac{\theta^t_i\strut^2}{\theta^t_{i+1}\theta^t_{i-1}}
\left(\frac1{u^t_i}+\delta u^t_{i-1}\right).
\end{align*}
This proves \eqref{eq:DT id}.

That $\tilde\theta^t_i$ is a solution of the transformed linear system follows by similar straightforward calculations. Let $\widehat u^t_i$ be the result of applying the Darboux transformation defined by $\tilde\theta^t_i$ to $\tilde u^t_i$. Then from \eqref{eq:DT tau},
\begin{equation}
\widehat\tau^t_i=\tilde\theta^t_i\tilde\tau^t_i=\tilde\theta^t_i\theta^t_i\tau^t_i=(1-\lambda^2)^{-i}(1-\lambda^2/\beta^2)^{t}\tau^t_i,
\end{equation}
and so, using \eqref{eq:u tau}, $\widehat u^t_i=u^t_i$, as required.
\end{proof} 

\begin{rem}\label{rem:discreteBS}
A generic solution $\theta^t_i$ of \eqref{eq:dlin1}, \eqref{eq:dlin2} has asymptotic form as $i\to\pm\infty$ given by \eqref{eq:infty} and \eqref{eq:-infty}.
Consequently, for $\tilde\theta^t_i$ given by \eqref{eq:1/theta}, 
\begin{equation}
  \tilde\theta_i^t\sim \frac{1}{a (1-\lambda)^{i}(1-\lambda/\beta)^{-t} + b (1+\lambda)^{i}(1+\lambda/\beta)^{-t}},
\end{equation}
as $i\to-\infty$ and
\begin{equation}
  \tilde\theta_i^t\sim \frac{1}{c (1-\lambda)^{i}(1-\lambda/\beta)^{-t} + d (1+\lambda)^{i}(1+\lambda/\beta)^{-t}},
\end{equation}
as $i\to+\infty$. When $0<\lambda<\beta<1$, the terms $(1-\lambda)^{i}$ and $(1+\lambda)^{i}$ tend to $\infty$ as $i\to-\infty$ and as $i\to+\infty$, respectively. Hence $\tilde\theta^t_i\to0$ as $i\to\pm\infty$, showing it has finite norm \eqref{def: dnorm} and therefore corresponds to a bound state. 
\end{rem}
As in the continuous case \cite{DeiftTrub1979}, a Darboux transformation determined by a generic eigenfunction $\theta_i^t$ can be used to add a bound state to $u_i^t$. That is, compared to $u_i^t$, the transformed potential $\tilde u_i^t$ will have an extra eigenvalue added to its discrete spectrum,  where we think of these solutions to the dKdV equation as potentials in spectral equation \eqref{eq:dlin1}. Such a  transformation is often referred to as a {\it dressing} transformation. Conversely, the Darboux transformation determined by the eigenfunction $\tilde\theta_i^t$ removes this bound state from $\tilde u_i^t$. A transformation based on a bound state eigenfunction is therefore referred to as an {\it undressing} transformation.

\begin{lem}\label{udDarboux}
Let $\Theta^t_i$ be a solution of linear problem \eqref{eq:lin1}--\eqref{eq:lin4}, and let the Darboux transformation \eqref{eq:ud darboux} determined by $\Theta^t_i$ map $U^t_i$ to $\tilde U^t_i$. Then
\begin{equation}
\label{eq:uDT id}
\max(\tilde U^t_{i-1}-1,-\tilde U^t_{i})=
2\Theta^t_i-\Theta^t_{i+1}-\Theta^t_{i-1}
+\max(U^t_{i-1}-1,-U^t_{i}).
\end{equation}
The transformed linear problem, \eqref{eq:lin1}--\eqref{eq:lin4} with potential $\tilde U^t_i$, has particular solution 
\begin{equation}\label{eq:u 1/theta}
\tilde\Theta^t_i:=\kappa i-\omega t-\Theta^t_i,   
\end{equation}
and the Darboux transformation defined by $\tilde\Theta^t_i$ maps $\tilde U^t_i$ back to $U^t_i$.
\end{lem}
\begin{proof}
Formula \eqref{eq:uDT id} is the ultradiscrete limit of \eqref{eq:DT id} and \eqref{eq:u 1/theta} is the ultradiscrete limit of \eqref{eq:1/theta}. A direct proof of \eqref{eq:uDT id} seems to be very difficult, but given that result, it is an entirely routine procedure to verify that $\tilde\Theta^t_i$ satisfies the transformed versions of \eqref{eq:lin1}--\eqref{eq:lin4}.
\end{proof}

The asymptotic behaviour, when $|i|\gg 1$ and $U_i^t=0$, of a generic solution to the ultradiscrete linear problem  can be easily deduced from the system
\begin{gather}
\label{eq:vaclin1}
  \max(\Phi^{t}_{i+1}-\kappa,\Phi^{t}_{i-1})=\Phi^{t}_{i},\\
\label{eq:vaclin2}  
 \max(\Phi^{t+1}_{i+1}-\kappa,\Phi^{t+1}_{i-1})=\Phi^{t+1}_{i},\\
\label{eq:vaclin3}  
  \max(\Phi^{t+1}_{i+1},\Phi^{t}_{i+1}-1)=\Phi^t_i,\\
 \label{eq:vaclin4}
  \max(\Phi^{t}_{i}+\kappa-\omega,\Phi^{t+1}_{i}+\kappa-1)=\Phi^{t+1}_{i+1}.
\end{gather}
Recall that $\omega$ is a free (non-negative) parameter and that $\kappa=\min(1,\omega)$. Let us consider the case where $\omega$ is not zero (and thus $\kappa>0$ as well). We  find from \eqref{eq:vaclin1} that either $\Phi_{i-1}^t = \Phi_i^t \,(\geq  \Phi_{i+1}^t-\kappa)$ or $\Phi_{i+1}^t-\kappa = \Phi_i^t \,(\geq \Phi_{i-1}^t)$, the former situation propagating naturally as $i\to -\infty$ and the latter as $i\to+\infty$. Hence, among the various possible asymptotic behaviours for a solution to \eqref{eq:lin1}--\eqref{eq:lin4}, there is one `generic' solution with asymptotic behaviour of the form
\begin{equation}\label{genericbehavphi}
\Phi_i^t \sim {\rm constant} \quad\text{as}~i\to-\infty\ \qquad\text{and}\qquad \Phi_i^t \sim \kappa i + \varphi(t) \quad\text{as}~i\to+\infty .
\end{equation}
A similar result follows from \eqref{eq:vaclin2} for $\Phi_i^{t+1}$ but equations \eqref{eq:vaclin3} and \eqref{eq:vaclin4} are clearly needed if one wants to understand how the asymptotics of $\Phi_i^{t}$ and $\Phi_i^{t+1}$ are related. Let us take the constant asymptotic value for $\Phi_i^{t}$ at $-\infty$ to be 0 (which is possible because of the linearity of equation \eqref{eq:lin1}). Equation \eqref{eq:vaclin3} then tells us that in the asymptotic regime where $i\to-\infty$ we have $\Phi_i^{t+1}=0$ as well. Equation \eqref{eq:vaclin4}, in this case, is trivially satisfied: $\max(\kappa-\omega, \kappa-1)=0 \Leftrightarrow \kappa=\min(1,\omega)$. On the other hand, when $i\to+\infty$, we find from \eqref{eq:vaclin4} that $\Phi_{i+1}^{t+1} \geq \kappa i + \varphi(t) + \kappa-\omega$ which implies that it must grow along with $\kappa i$. Hence the generic asymptotic behaviour for $\Phi_{i}^{t+1} $ is exactly that of \eqref{genericbehavphi} for some value of $\varphi(t+1)$, to be determined. Setting $\Phi_i^{t+1}=\kappa i + \varphi(t+1)$ in \eqref{eq:vaclin4} as $i\to+\infty$ we obtain
\begin{equation}
\max(\varphi(t) -\omega, \varphi(t+1) -1) = \varphi(t+1),
\end{equation}
which implies that $\varphi(t+1)=\varphi(t) - \omega$ (note that in this case equation \eqref{eq:vaclin3} is trivially satisfied). We therefore say that a {\em generic} solution of the ultradiscrete linear problem \eqref{eq:lin1}--\eqref{eq:lin4} has asymptotic behaviour
\begin{equation}\label{udgeneig-asympt}
\Phi_i^t \sim \left\{\begin{array}{lr} 0 & ~(i\to-\infty) \\ \kappa ~(i - \phi - c t) & (i\to+\infty)\end{array}\right. ,
\end{equation}
for some constant $\phi$ and where $c=\max(1,\omega)$.

In the next section (Section \ref{sec:line sol})  we will show how to explicitly describe generic solutions to \eqref{eq:lin1}--\eqref{eq:lin4} and how such solutions can be seen to give rise to dressing transformations for the potential in the ultradiscrete linear problem.

\begin{rem}\label{rem:BSdef}
By analogy with the discrete case, if $\Theta_i^t$ of Lemma \ref{udDarboux} is a generic solution to the ultradiscrete linear problem \eqref{eq:lin1}--\eqref{eq:lin4}, $\tilde\Theta_i^t$ can be regarded as the ultradiscrete counterpart of a bound state $\tilde\theta^t_i$ and will be referred to as a {\it bound state} eigenfunction as well, even though it does not tend to zero asymptotically. In fact, in Section \ref{sec:bound states} we will show that such $\tilde\Theta_i^t$ indeed have many properties in common with their discrete counterparts. In particular they can be shown to define undressing transformations for potentials in the ultradiscrete linear problem \eqref{eq:lin1}--\eqref{eq:lin4}.
\end{rem}

\section{Ultradiscrete eigenfunctions and dressing transformations}\label{sec:line sol}
We will show in this section that for any $U^t_i$ the ultradiscrete linear system \eqref{eq:lin1}--\eqref{eq:lin4} has two solutions which may be combined (using $\max$) to give a two parameter family of solutions. One of these parameters is the free parameter $\omega\ge0$ of the linear system. As part of the verification of these solutions we will see that each solution exists if and only if $\omega\ge\omax$ where $\omax$ is the mass of the heaviest soliton in $U^t_i$. As in the previous section, it is sometimes useful to distinguish two cases, $V^t\le1$ and $V^t>1$,  for $V^t$ as in \eqref{eq:V^t}. Recall that in the former case, $V^t$ is independent of $t$ and $\omax=V^t$ and in the latter case, $\omax>1$.

\subsection{Solution $\Phi^t_i=\sum_{j<i}U^{t-1}_j$}\label{sec:basicsol1}
{Note that, using \eqref{eq:naive},  \eqref{eq:lin2} can be seen to be the time update of \eqref{eq:lin1},  consequently} we only need to check \eqref{eq:lin1}, \eqref{eq:lin3} and \eqref{eq:lin4}. For any $t$, substituting $\Phi^t_i=\sum_{j<i}U^{t-1}_j$ in \eqref{eq:lin1} gives
\begin{align}
\max\left(\sum_{j<i+1}U^{t-1}_j-\kappa,\sum_{j<i-1}U^{t-1}_j\right)&=\sum_{j<i}U^{t-1}_j+\max(U^{t}_{i-1}-1,-U^{t}_{i})\nonumber
\intertext{and so, using \eqref{eq:naive} to rewrite the RHS,}
\label{eq:sol1 (16)}
\max(U^{t-1}_{i}+U^{t-1}_{i-1}-\kappa,0)&=\max(U^{t-1}_{i}+U^{t-1}_{i-1}-1,0).
\end{align}
We consider two cases: if $V^t\le1$, $V^{t-1}=V^t=\omax\le1$ (cf. Remark \ref{rem:omega max}) and so \eqref{eq:sol1 (16)} is equivalent to $\max(U^{t-1}_{i}+U^{t-1}_{i-1}-\kappa,0)=0$, that is  $U^{t-1}_{i}+U^{t-1}_{i-1}\le\kappa$. {Since $\kappa=\min(1,\omega)$, in this case \eqref{eq:lin1} is satisfied if and only if $\omega\ge\kappa\ge\omax$.} In the case $V^t>1$, we have $V^{t-1}>1$ and $\omax>1$. For some $j$, $\max(U^{t-1}_{j}+U^{t-1}_{j-1}-\kappa,0)=U^{t-1}_{j}+U^{t-1}_{j-1}-1>0$, giving $\kappa=1$. In this case \eqref{eq:lin1} is satisfied if and only if $\omega\ge1$. In both cases, the solution of \eqref{eq:lin1} is verified (at least) when $\omega\ge\omax$.

Next, substituting $\Phi^t_i=\sum_{j<i}U^{t-1}_j$ in \eqref{eq:lin3} gives
\begin{align*}
\max\left(\sum_{j<i+1}U^{t}_j,\sum_{j<i+1}U^{t-1}_j+U^t_i-1\right)&=\sum_{j<i}U^{t-1}_j,
\end{align*}
which is equivalent to
\begin{align*}
\min\left(\sum_{j<i+1}(U^{t-1}_j-U^{t}_j),1-U^t_i\right)=\min\left(\sum_{j>i}(U^{t}_j-U^{t-1}_j),1-U^t_i\right)&=U^{t-1}_i,
\end{align*}
using conservation of mass \eqref{eq:cons U}. This is simply the downdate rule \eqref{eq:downdate rule} and hence \eqref{eq:lin3} is satisfied without restriction on $\omega$.

Finally, \eqref{eq:lin4} gives
\begin{align*}
\max\left(\sum_{j<i}U^{t-1}_j+\kappa-\omega,\sum_{j<i}U^t_j+U^t_i+\kappa-1\right)&=\sum_{j<i+1}U^t_j,
\intertext{and so}
\max\left(1+\sum_{j<i}U^{t-1}_j-\sum_{j\le i}U^t_j-\omega,0\right)&=1-\kappa.
\end{align*}
which may be written as
\begin{equation}\label{eq:sol1 (19)}
\max\left(1+Y^t_i-\omega,0\right)=\max(1-\omega,0),
\end{equation}
where $Y^t_i=\sum_{j>i}U^t_j-\sum_{j\ge i}U^{t-1}_j$ is the conserved density defined in \eqref{eq:Y}. If $V^t\le1$, then $Y^t_i=0$ for all $i$ (Lemma~\ref{lem:V^t}) and so the condition \eqref{eq:sol1 (19)} holds without restriction. Else $V^t>1$ and for all $i$, $1+Y^t_i\le\omax$ (Theorem~\ref{thm:cons X Y}) {since} $\omax>1$ (Corollary~\ref{cor:newcor}). If $\omega\ge\omax>1$ then \eqref{eq:sol1 (19)} is {trivially} satisfied for all $i$. On the other hand, if $\omega<\omax$, for some $j$, $1+Y^t_j = 1 + \max Y_i^t =\omax$ and \eqref{eq:sol1 (19)} {gives $\omax=\max(1,\omega)$, which means that $\omax$ either takes the value $1$ or $\omega$, both values leading to a contradiction.} Hence \eqref{eq:sol1 (16)} is satisfied for all $i$ if and only if $\omega\ge\omax$. 

It follows that the system of ultradiscrete linear equations \eqref{eq:lin1}--\eqref{eq:lin4} is satisfied by $\Phi^t_i=\sum_{j<i}U^{t-1}_j$ for all $i$ if and only if $\omega\ge\omax$.

\subsection{Solution $\Phi^t_i=\kappa i-\omega t+\sum_{j\ge i}U^t_j$}\label{sec:basicsol2}
The verification of this second solution is very similar to that of the first. For any $t$, substituting $\Phi^t_i=\kappa i-\omega t+\sum_{j\ge i}U^t_j$ in \eqref{eq:lin1} gives
\begin{align*}
\max\left(\kappa i-\omega t+\sum_{j\ge i+1}U^{t}_j,\kappa (i-1)-\omega t+\sum_{j\ge i-1}U^{t}_j\right)&=
\kappa i-\omega t+\sum_{j\ge i}U^{t}_j+\max(U^{t}_{i-1}-1,-U^{t}_{i})\nonumber
\intertext{and therefore,}
\max(U^{t}_{i-1}+U^{t}_{i}-\kappa,0)&=\max(U^{t}_{i-1}+U^{t}_{i}-1,0),
\end{align*}
and so by the argument used following \eqref{eq:sol1 (16)}, it is sufficient that $\omega\ge\omax$.

Next, substituting $\Phi^t_i=\kappa i-\omega t+\sum_{j\ge i}U^t_j$ in \eqref{eq:lin3} gives
\begin{align*}
\max\left(\kappa-\omega+\sum_{j\ge i+1}U^{t+1}_j,\kappa + \sum_{j\ge i}U^{t}_j-1\right)&=\sum_{j\ge i}U^{t}_j,
\end{align*}
which by similar algebra as in the case of the first solution can be seen to be equivalent to
\begin{align*}
\max\left(1+Y_i^{t+1}-\omega,0\right)=\max(1-\omega,0).
\end{align*}
Since, in case $V^t\le1$ we have that $V^{t+1}=V^t\le1$ and therefore $Y_j^{t+1}=0$, and since $V^t>1$ implies that $V^{t+1}>1$ as well, exactly the same reasoning as for \eqref{eq:sol1 (19)} leads to the conclusion that equation \eqref{eq:lin3} is satisfied if and only if $\omega\ge \omax$.

Substituting $\Phi^t_i=\kappa i-\omega t+\sum_{j\ge i}U^t_j$ in \eqref{eq:lin4} gives
\begin{align*}
\max\left(\sum_{j\ge i} U_j^t, \sum_{j\ge i} U_j^{t+1}+ U_i^t-1\right)=\sum_{j\ge i+1}U_j^{t+1},
\end{align*}
which is equivalent to
\begin{align*}
U_i^{t+1}=\min\left(\sum_{j\ge i}(U_j^{t+1}-U_j^t), 1-U_i^t\right)=\min\left(1-U_i^t, \sum_{j< i}(U_j^{t}-U_j^{t+1}) \right),
\end{align*}
which is just the update rule \eqref{eq:update rule} and hence equation \eqref{eq:lin4} is automatically satisfied.

In summary, we find that the requirement for $\Phi^t_i=\kappa i-\omega t+\sum_{j\ge i}U^t_j$ to satisfy the system  \eqref{eq:lin1}--\eqref{eq:lin4}  is again that $\omega\ge\omax$.

\begin{thm}\label{thm:gen sol}
For any $U^t_i$, the linear system \eqref{eq:lin1}--\eqref{eq:lin4} has solutions
\begin{equation}
  \Phi^t_i=\sum_{j<i}U^{t-1}_j\text{ and \ } \Phi^t_i=\kappa i-\omega t+\sum_{j\ge i}U^t_j,
\end{equation}
where $\omega\ge \omax$. The max-linear combination of these gives a two parameter $(\omega,\phi)$ family of solutions
\begin{equation}\label{eq:max sol}
  \Theta^t_i=\max\left(\sum_{j<i}U^{t-1}_j,\kappa (i-\varphi^t)+\sum_{j\ge i}U^t_j\right),
\end{equation}
where $\kappa=\min(1,\omega)$, $\varphi^t=\phi+ct$ and $c=\max(1,\omega)$. Moreover, there exist some $m^t\in\mathbb Z$, dependent on $t$, such that 
\begin{equation}\label{eq:max sol alt}
  \Theta^t_i=
  \begin{cases}
    \displaystyle\sum_{j<i}U^{t-1}_j&i<m^t\\
    \displaystyle\kappa (i-\varphi^t)+\sum_{j\ge i}U^t_j&i\ge m^t.
  \end{cases}
\end{equation}
\end{thm}
\begin{proof}
The basic solutions were verified in the preceding paragraphs and the fact that their maximum is also a solution follows from the general properties of $\max$. 

Define the difference of the two basic solutions 
\begin{equation}\label{eq:F}
F^t_i=\kappa(i-\varphi^t)+\sum_{j\ge i}U^{t}_j-\sum_{j<i}U^{t-1}_j.
\end{equation}
Since $\omega\ge\omax$ and so $\kappa\ge\kmax$, \eqref{eq:basic kmax} gives
\begin{equation}\label{eq:F inc}
F^t_{i+1}-F^t_i=\kappa-U^t_i-U^{t-1}_i\ge\kmax-U^t_i-U^{t-1}_i\ge0,
\end{equation}  
and so the sequence $F^t_i$ is weakly increasing in $i$. Notice that $F_i\to-\infty$ as $i\to-\infty$.  The alternative form of the solution \eqref{eq:max sol alt} then follows immediately. 
The {\em split point} $m^t$ is any of the (consecutive) integers for which $F^t_{m^t}$ attains its smallest nonnegative value.
\end{proof}

{Notice that the solution \eqref{eq:max sol alt} has exactly the asymptotic behaviour \eqref{udgeneig-asympt},
\begin{equation}
\Theta_i^t \sim \left\{\begin{array}{lr} 0 & ~(i\to-\infty) \\ \kappa ~(i - \varphi^t) & (i\to+\infty)\end{array}\right. ,
\end{equation}
and we conclude that with equations \eqref{eq:max sol alt} (or equivalently \eqref{eq:max sol})  we have found an explicit expression for a generic solution to the ultradiscrete linear problem. Besides the spectral parameter $\omega$ this solution has one more free parameter: a phase constant $\phi$ which fixes the asymptotic behaviour.}

\begin{exa}\label{exa:genericeigf}
This example illustrates the construction of a {such an} ultradiscrete eigenfunction. Consider the ultradiscrete state $U^0_i=\dots,0, 0, 0,0,0,1,\frac12,0,1,0,0,0,0,0,0\dots$ where the index of the first zero is 1. We observe that $\omax=\frac32$ and we choose $\omega=2\ge\omax$, giving $\kappa=1$, and choose $\phi=7$. We obtain
\[
\setlength{\arraycolsep}{0pt}
\renewcommand{\arraystretch}{1.5}
  \begin{array}{rCCCCCCCCCCCCCCCCCCCCCCCCCCCCCCCCCCCCCC}
U^0_i\ :&&0&0&0&0&0&1&\frac12&0&1&0&0&0&0&0&0\\
U^{-1}_i\ :&&0&0&0&\frac12&1&0&0&1&0&0&0&0&0&0&0\\
i-7+\sum_{j\ge i}U^{0}_j\ :&&\overline{\tfrac72}&\overline{\tfrac52}&\overline{\frac32}&\overline{\frac12}&\frac12&\frac32&\frac32&2&3&3&4&5&6&7&8\\
\sum_{j<i}U^{-1}_j\ :&&0&0&0&0&\frac12&\frac32&\frac32&\frac32&\frac52&\frac52&\frac52&\frac52&\frac52&\frac52&\frac52\\
F^0_i\ :&&\overline{\frac72}&\overline{\frac52}&\overline{\frac32}&\overline{\frac12}&0&0&0&\frac12&\frac12&\frac12&\frac32&\frac52&\frac72&\frac92&\frac{11}2
\end{array}
\]
The split point $m^0$ may be chosen to be any of the indices 5, 6, or 7 (giving the smallest nonnegative values of $F^0_i$) and then the ultradiscrete eigenfunction at $t=0$ is
\begin{equation}
\Theta^0_i=\max\left(i-7+\sum_{j\ge i}U^{0}_j,\sum_{j<i}U^{-1}_j\right)=\dots,0,0,0,0,\tfrac12,\tfrac32,\tfrac32,2,3,3,4,5,6,7,8,\dots.
\end{equation}

Now consider $t=1$:
\[
\setlength{\arraycolsep}{0pt}
\renewcommand{\arraystretch}{1.5}
  \begin{array}{rCCCCCCCCCCCCCCCCCCCCCCCCCCCCCCCCCCCCCC}
U^1_i\ :&&0&0&0&0&0&0&\frac12&1&0&1&0&0&0&0&0\\
U^{0}_i\ :&&0&0&0&0&0&1&\frac12&0&1&0&0&0&0&0&0\\
i-9+\sum_{j\ge i}U^{1}_j\ :&&\overline{\frac{11}2}&~\overline{\frac92}~&\overline{\frac72}&\overline{\frac52}&\overline{\frac32}&\overline{\frac12}&\frac12&1&1&2&2&3&4&5&6\\
\sum_{j<i}U^{0}_j\ :&&0&0&0&0&0&0&1&\frac32&\frac32&\frac52&\frac52&\frac52&\frac52&\frac52&\frac52\\
F^1_i\ :&&\overline{\frac{11}2}&~\overline{\frac92}~&\overline{\frac72}&\overline{\frac52}&\overline{\frac32}&\overline{\frac12}&\overline{\frac12}&\overline{\frac12}&\overline{\frac12}&\overline{\frac12}&\overline{\frac12}&\frac12&\frac32&\frac52&\frac72
\end{array}
\]
At this time, the split point $m^1$ may only be chosen to be index 12 {(as this gives the smallest nonnegative value of $F_i^1$)}
and the ultradiscrete eigenfunction at $t=1$ is
\begin{equation}
\Theta^1_i=\max\left(i-9+\sum_{j\ge i}U^{1}_j,\sum_{j<i}U^{0}_j\right)=\dots,0,0,0,0,0,0,1,\tfrac32,\tfrac32,\tfrac52,\tfrac52,3,4,5,6,\dots.
\end{equation}

\end{exa}

\subsection{Generic eigenfunctions and dressing transformations}\label{sec:dressing}
We can now explain the connection between Darboux transformations, defined in terms of generic ultradiscrete eigenfunctions $\Theta_i^t$ as given by \eqref{eq:max sol} or \eqref{eq:max sol alt}, and the work of Nakata \cite{MR2545618} who first introduced a B\"acklund transformation for the udKdV equation that acts as a dressing transformation.

The effect of a Darboux transformation on an ultradiscrete tau function $T^t_i$ and its corresponding $U^t_i$ is stated in \eqref{eq:ud darboux}. With $\Theta^t_i$ given by \eqref{eq:max sol} we have
\begin{equation}
\tilde T^t_i=T^t_i+\max\Big(\sum_{j<i}U^{t-1}_j,\kappa(i-\varphi^t)+\sum_{j\ge i}U^{t}_j\Big),
\label{eq:sym1}
\end{equation}
for $\omega\ge\omax$ and $\kappa=\min(1,\omega)>0$.

Now we use the discrete potential $Z^t_i:=T^t_i-T^{t+1}_i$ for $U^t_i$ we introduced in Section \ref{sec:sol prop}, such that $U^t_i=\Delta Z^t_i$ (cf. \eqref{eq:ud H sub}), and we impose the boundary condition $Z_-=-\frac12\sum_{j\in\mathbb{Z}}U_j^t$ for sufficiently large $-i$, for all $t$. 
This allows us to rewrite the (in practice, finite) sum $\sum_{j<i}U^{t}_j$ in terms of $T$-functions:
\begin{align}
\sum_{j<i}U^{t}_j=-Z_- + Z^t_i= \frac12\sum_{j\in\mathbb{Z}}U_j^t + T^{t}_i-T^{t+1}_i.
\end{align}
Using this expression at $t$ and $t-1$ in \eqref{eq:sym1}, we find
\begin{align}\nonumber
\tilde T^t_i&=T^t_i+\max\Big(\sum_{j<i}U^{t-1}_j,\kappa(i-\varphi^t)+\sum_{j\in\mathbb{Z}}U_j^t-\sum_{j< i}U^{t}_j\Big) \\
&= \frac12\sum_{j\in\mathbb{Z}}U_j^t + \tfrac12\kappa(i-\varphi^t)+\max\Big(\tfrac12\kappa(i-\varphi^t)+T^{t+1}_i,-\tfrac12\kappa(i-\varphi^t)+T^{t-1}_i\Big),
\end{align}
and using the gauge freedom for the $T$-functions ($\tilde T^t_i\sim \tilde T^t_i+\alpha i+\beta t+\gamma$ for appropriate constants $\alpha,\beta,\gamma$) we obtain an equivalent formula for the transformed $T$-function
\begin{align}
\label{eq:nakata}
\tilde T_i^t\sim\max\Big(\tfrac12\kappa(i-\varphi^t)+T^{t+1}_i,-\tfrac12\kappa(i-\varphi^t)+T^{t-1}_i\Big),
\end{align}
which is precisely the formula for the vertex operator form of the B\"acklund transformation for udKdV given in the BBS case \cite{MR2545618} (or real-valued case \cite{MR2738130}). In \cite{MR2545618}, \eqref{eq:nakata} was shown to transform $T$-functions that satisfy the ultradiscrete bilinear KdV equation  \eqref{eq:H udKdV} to functions $\tilde T$ that again satisfy the  same equation, provided that $\omega$ is not less than the maximum soliton mass contained in $U_i^t$.   

It is interesting to calculate the asymptotic value of the discrete potential $\tilde Z_i^t$ for the $T$-function given by \eqref{eq:nakata}. As explained in Section \ref{sec:sol prop}, for $i$ (or $-i$) sufficiently large, the asymptotic values of $Z_i^t$ no longer depend on $t$ and we can set $T_i^t - T_i^{t+1} = Z_\pm$ in $\tilde Z_i^t=\tilde T_i^t - \tilde T_i^{t+1}$ for \eqref{eq:nakata}, in both asymptotic regimes:
\begin{equation}
|i|\gg1 :\qquad \tilde Z_i^t = \max\Big(\tfrac12\kappa(i-\varphi^t)-Z_\pm,-\tfrac12\kappa(i-\varphi^t)+Z_\pm\Big)-\max\Big(\tfrac12\kappa(i-\varphi^{t+1})-2Z_\pm,-\tfrac12\kappa(i-\varphi^{t+1})\Big),
\end{equation}
from which we find that
\begin{equation}
\tilde Z_+ = Z_+ + \frac{\omega}{2}\qquad \text{and} \qquad \tilde Z_- = Z_- - \frac{\omega}{2}.
\end{equation}
Moreover, since the difference $Z_+-Z_-$ is equal to the total mass $\sum_{i\in\mathbb{Z}} U_i^t$ of the solution $U_i^t$ obtained from the discrete potential $Z_i^t$ (cf. Section \ref{sec:sol prop}), we immediately find that the dressing transformation has increased the mass of this solution by $\omega$:
\begin{equation}\label{eq:massloss}
\sum_{i\in\mathbb{Z}} \tilde U_i^t = \tilde Z_+ - \tilde Z_- = Z_+ - Z_- + \omega = \sum_{i\in\mathbb{Z}} U_i^t + \omega.
\end{equation}
Since we chose the gauge of the initial $T$-function such that $Z_-=-\frac12\sum_{i\in\mathbb{Z}}U_i^t$, this also shows that the transformed $T$-function given by \eqref{eq:nakata} preserves this boundary condition: $\tilde Z_-=-\frac12\sum_{i\in\mathbb{Z}}\tilde U_i^t$. When iterating the dressing Darboux transformation for increasing values of $\omega$, the link with Nakata's B\"acklund transformation will therefore always be exactly as explained above. 

Using the formula \eqref{eq:max sol alt} for $\Theta^t_i$ instead, \eqref{eq:sym1} becomes
\begin{equation}\label{eq:tilde T split}
\tilde T^t_i = \frac12\sum_{j\in\mathbb{Z}}U_j^t +
\begin{cases}
T^{t-1}_i&i<m^t\\
\kappa(i-\varphi^t)+T^{t+1}_i&i\ge m^t.  
\end{cases}
\end{equation}
As discussed in connection to Theorem \ref{thm:gen sol}, the split point $m^t$ is defined to be any of the integers for which the difference $F^t_{m^t}$ \eqref{eq:F} of the two basic solutions of Sections \ref{sec:basicsol1} and  \ref{sec:basicsol2} attains its smallest nonnegative value. The split point is thus not only $t$-dependent, but also depends on the phase parameter $\phi$.
This is a complicated implicit definition, since the limits in the summations in the formula for $F^t_i$ depend on $i$. We are unable to solve in general for $m^t$, in terms of $\phi$, although it is not difficult to compute $m^t$ in any given example. 
However, as for large $|i|$ all $T$-functions in \eqref{eq:ud H sub} are given by the same clause in \eqref{eq:tilde T split}, we can describe the transformed solution for large enough $|i|$:
\begin{equation}\label{eq:tilde U split}
\tilde U^t_i=
\begin{cases}
U^{t-1}_i&i<\min(m^{t-1},m^{t},m^{t+1})-1\\
U^{t+1}_i&i\ge\max(m^{t-1},m^{t},m^{t+1}).  
\end{cases}
\end{equation}
Hence we see that, roughly speaking, the effect of the dressing Darboux transformation is to downdate $U^t_i$ in the left part and update it in the right part in order to create a space for the new soliton, with mass $\omega$, to be inserted. This, by the way, shows that a dressing transformation of a state with finite support again yields a state with finite support. 

\begin{rem}
Notice that the fact that $\widetilde U_i^t$ has finite support, combined with the fact that the dressing \eqref{eq:nakata} maps udKdV tau functions to tau functions (as shown in \cite{MR2545618}) proves that a dressing Darboux transformation, using a generic solution to \eqref{eq:lin1}--\eqref{eq:lin4}, indeed maps solutions to the udKdV equation to solutions.
\end{rem}

\begin{exa}\label{exa:dressing}
Let us calculate the dressing of the state $U^0_i=\dots,0, 0, 0,0,0,1,\frac12,0,1,0,0,0,0,0,0,\dots$ (for which $\omax=3/2$) with the generic eigenfunction $\Theta_i^t$ we calculated for it in Example \ref{exa:genericeigf}, in the case of $\omega=2$ and phase constant $\phi=7$. Recall that the dressed state is given by $\widetilde U^0_i=U^0_i+\Theta^0_{i+1}+\Theta^{1}_i-\Theta^0_{i}-\Theta^{1}_{i+1}$. We have,
\[
\setlength{\arraycolsep}{0pt}
\renewcommand{\arraystretch}{1.5}
  \begin{array}{rCCCCCCCCCCCCCCCCCCCCCCCCCCCCCCCCCCCCCC}
U^0_i\ :&&0&0&0&0&0&1&\frac12&0&1&0&0&0&0&0&0\\
\Theta^{0}_i\ :&&0&0&0&0&\tfrac12&\tfrac32&\tfrac32&2&3&3&4&5&6&7&8\\
\Theta^{1}_i\ :&&0&0&0&0&0&0&1&\tfrac32&\tfrac32&\tfrac52&\tfrac52&3&4&5&6\\
\widetilde U_i^0\ :&&0&0&0&\frac12&1&0&\frac12&1&0&1&\frac12&0&0&0&0.
\end{array}
\]
That $\widetilde U_i^0$, compared to $U_i^0$,  has indeed gained a 2-soliton can be seen immediately on the time evolved states,
\[
\setlength{\arraycolsep}{0pt}
\renewcommand{\arraystretch}{1.5}
  \begin{array}{rCCCCCCCCCCCCCCCCCCCCCCCCCCCCCCCCCCCCCC}
\widetilde U_i^0\ :&&0&0&0&\frac12&1&0&\frac12&1&0&1&\frac12&0&0&0&0\\
\widetilde U^{1}_i\ :&&0&0&0&0&0&1&\tfrac12&0&1&0&\tfrac12&1&\tfrac12&0&0\\
\widetilde U^{-1}_i\ :&&0&\tfrac12&1&\tfrac12&0&1&\tfrac12&0&1&0&0&0&0&0&0,
\end{array}
\]
while  further analysis shows that, at $t=0$, the inserted 2-soliton is in full interaction with the other solitons:
\[
\setlength{\arraycolsep}{0pt}
\renewcommand{\arraystretch}{1.5}
  \begin{array}{rCCCCCCCCCCCCCCCCCCCCCCCCCCCCCCCCCCCCCC}
\widetilde U^{-1}_i-\widetilde U^{1}_i\ :&&0&\tfrac12&1&\tfrac12&0&0&0&0&0&0&\overline{\tfrac12}&\overline{1}&\overline{\tfrac12}&0&0\\
\widetilde X^0_i\ :&&0&0&\tfrac12&\tfrac32&2&2&2&2&2&2&2&\tfrac32&\tfrac12&0&0 ,
\end{array}\hskip.85cm
\]
where $\widetilde X_i^0 = \sum_{j< i}(\widetilde U^{-1}_j-\widetilde U^{1}_j)$. 
\end{exa}

In fact, the evolved states $\tilde U_i^{-1}$ and $\tilde U_i^1$ in the above example also show that the solitons that were already present in the original state $U_i^0$ are unaltered in the dressing transformation: asymptotically they are shifted in phase, but they are all still present, with their original masses intact. In order to explain why this is true in general, we first need to construct a transformation that actually reverses a dressing: a so-called undressing transformation.

\section{Undressing transformations and ultradiscrete bound states}\label{sec:bound states}

We have already seen in Lemma \ref{udDarboux} that to every dressing Darboux transformation from $U^t_i$ to $\tilde U^t_i$, defined by a generic eigenfunction $\Theta^t_i$ through \eqref{eq:ud darboux}, there corresponds an undressing transformation from $\tilde U^t_i$ to $U^t_i$, defined by $\tilde\Theta^t_i$, given in \eqref{eq:u 1/theta}. If the parameter $\omega>0$, then the dressing transformation adds a soliton of that mass and the undressing transformation will therefore again remove that soliton. 

In the previous section we obtained the expression 
\begin{equation}
\Theta^t_i=\max\Big(\kappa(i-\varphi^t)+\sum_{j\ge i}U^t_j,\sum_{j<i}U^{t-1}_j\Big),
\end{equation}
and, through \eqref{eq:u 1/theta}, we obtain
\begin{equation}\label{eq:exp tilde Th}
\tilde\Theta^t_i=\kappa \phi + \min\Big(\kappa(i-\varphi^t)-\sum_{j<i}U^{t-1}_j,-\sum_{j\ge i}U^t_j\Big),
\end{equation}
which, according to Lemma \ref{udDarboux} and Remark \ref{rem:BSdef}, is a bound state eigenfunction for the linear system \eqref{eq:lin1}--\eqref{eq:lin4} for the dressed solution $\tilde U _i^t$. 
However, for the purpose of obtaining an explicit undressed state $U^t_i$ from a given state $\tilde U^t_i$, the expression \eqref{eq:exp tilde Th} is of no use: it is not defined in terms of the known $\tilde U^t_i$ values, rather it is defined in terms of the unknown target values $U^t_i$. 
In this section we will obtain an alternative formula for this eigenfunction $\tilde\Theta^t_i$ expressed in terms of the known initial potential. From now on however we shall dispense with the notation $\ \tilde\ \ $ for the initial  state to which we wish to apply the undressing transformation, to stress that our construction is fully general and  does not rely on any prior dressing that might or might not have taken place. 

As explained in Remark \ref{newposremark5}, when $V^t=0$ there are no solitons and, as was done in Section \ref{sec:sol prop}, we can write an explicit solution to the ultradiscrete KdV equation (explicit in $i$ and $t$) for any such given state. In this section we shall therefore always assume that the state $U_i^t$ we want to apply the undressing transformation to is such that $V^t>0$ (or, equivalently, that $\omax>0$) and, as before, that it has finite support.

\subsection{Ultradiscrete bound state eigenfunctions}\label{sec:BSeig}
Consider the following $\min$-linear combination of the basic solutions to \eqref{eq:lin1}--\eqref{eq:lin4} that we discussed in sections \ref{sec:basicsol1} and \ref{sec:basicsol2}, for a given potential $U_i^t$:
\begin{equation}\label{eq:bar Theta}
\bar\Theta^t_i:=\min\Big(\kappa(i-\varphi^t)+\sum_{j\ge i} U^{t}_j,\sum_{j<i} U^{t-1}_j\Big).
\end{equation}
Note that the asymptotic form of the function $\bar\Theta^t_i$,
\begin{equation}\label{eq:asympt}
\bar\Theta^t_i\sim \sum_{j\in\mathbb Z} U^t_j +
\begin{cases}
\kappa(i-\varphi^t)&\text{as }i\to-\infty\\
0&\text{as }i\to\infty,  
\end{cases}
\end{equation}
is essentially the same as that for the bound state \eqref{eq:exp tilde Th}, up to an inconsequential renormalisation of the latter by $\sum_{j\in\mathbb Z} U^t_j-\kappa\phi$. It will turn out that up to this renormalisation, both functions are actually identical for all $i$ and $t$.

Let us first prove that $\bar \Theta_i^t$ also satisfies  the linear system \eqref{eq:lin1}--\eqref{eq:lin4} for $U_i^t$.
It has already been shown in \eqref{eq:F inc} that when $\omega\geq\omax$ the difference $ F^t_i$ of the two basic solutions, given by \eqref{eq:F},  is weakly increasing in $i$ and so we may express \eqref{eq:bar Theta} as
\begin{equation}\label{eq:bar Theta split}
\bar \Theta^t_i=
\begin{cases}
\displaystyle\kappa(i-\varphi^t)+ \sum_{j\ge i} U^{t}_j&i<m^t\\
\displaystyle\sum_{j<i} U^{t-1}_j&i\ge m^t
\end{cases}
\end{equation}
for some $m^t\in\mathbb Z$. Also, we compute the difference in $t$,
\begin{equation}\label{eq:diff F}
F^{t+1}_{i}- F^t_i=\sum_{j\ge i}( U^{t+1}_j- U^{t-1}_j)-\omega= X^t_i-\omega,
\end{equation}
with $X_i^t$ as in \eqref{eq:X} and where we have used the conservation of the total mass \eqref{eq:cons U}. 

From Theorem~\ref{thm:cons X Y} we have seen that $\omax=\max_i X^t_i$. Now we choose $\omega=\omax$ and $m^t$ to be any index at which this maximum is attained. In other words, $m^t$ is chosen such that $ X^t_{m^t}=\omax$. Then we have
\begin{equation}\label{eq:diff F at m}
F^{t+1}_{m^t}= F^t_{m^t}.
\end{equation}

Notice that this definition of the split point $m^t$ differs from that in \eqref{eq:max sol alt}, in Theorem \ref{thm:gen sol}. We shall see however that there exists a special choice for the phase constant $\phi$ such that $\bar\Theta_i^t$ indeed becomes a solution to the linear system  \eqref{eq:lin1}--\eqref{eq:lin4}.

The set
\begin{equation}\label{def:Mt}
M^t=\{i\in\mathbb Z: X^t_i=\omax\}
\end{equation}
may be written as the union $M^t=M^t_1\cup M_2^t\cup\cdots\cup M^t_k$ of one or more sets of consecutive integers, corresponding to disconnected peaks or plateaux of maximal height in the graph of $X^t_i$, plotted as a function of $i$ at fixed $t$. As discussed in Remark~\ref{rem:omega max}, these indicate the location of one or more solitons of maximal mass. For example, in the first example in Examples~\ref{exa:omax} in Section~\ref{sec:conserved}, we find $\omax=3$ with $M^t=\{5,6,7,8,9,10\}\cup\{14,15,16,17\}$, which indicates the presence of (at least) two solitons with mass 3 and in the second example we find $\omax=2$ with $M^t=\{5,6,7,8,9,10,11\}$ indicating the existence of at least one soliton with mass 2. In fact, in this example there are two mass-2 solitons.

Next, since $V^t>0$ (and hence $\omax>0$) we have $\kmax>0$ and we can choose the phase constant $\phi=\phi_{\text{max}}$ to be such that $F^t_{m^t}=0$, giving 
\begin{equation}\label{eq:varphi}
  \varphi^t_{\text{max}}=m^t+\dfrac{1}{\kmax}\Big(\sum_{j\ge m^t} U^{t}_j-\sum_{j<m^t} U^{t-1}_j\Big)=m^t+\dfrac{1}{\kmax}\Big(\sum_{j\ge m^t} U^{t-1}_j-\sum_{j<m^t} U^{t}_j\Big),
\end{equation}
and so the phase constant is 
\begin{equation}\label{eq:phi}
\phi_{\text{max}}=m^t-c_{\text{max}}t+\dfrac{1}{\kmax}\Big(\sum_{j\ge m^t} U^{t-1}_j-\sum_{j<m^t} U^{t}_j\Big).
\end{equation}
Using \eqref{eq:diff F at m} it is easily confirmed that this expression---despite its appearance---is indeed independent of $t$.

Note that because we chose $\phi_{\max}$ such as to have $F^{t}_{m^t}=0$, the split point $m^t$ now plays exactly the same role as for the generic solution in Theorem~\ref{thm:gen sol}. Furthermore, as $F^{t}_{m^t}=F^{t+1}_{m^t}=0$, we may choose the split point $m^{t+1}$ in the definition of $\bar\Theta^{t+1}_i$ to be $m^t$, as for $\bar\Theta^{t}_i$. This \emph{does not} however imply, in general, that $m^t\in M^{t+1}$. 
Thus, in summary, we have shown that
\begin{equation}\label{eq:overline Theta}
\bar \Theta^t_i=
\begin{cases}
\displaystyle\kmax(i-\varphi^t_{\text{max}})+\sum_{j\ge i} U^{t}_j&i<m^t\\
\displaystyle\sum_{j<i} U^{t-1}_j&i\ge m^t,
\end{cases}
\end{equation}
and 
\begin{equation}\label{eq:overline Theta +1}
\bar \Theta^{t+1}_i=
\begin{cases}
\displaystyle\kmax(i-\varphi^t_{\text{max}})-\omax+\sum_{j\ge i} U^{t+1}_j&i<m^t\\
\displaystyle\sum_{j<i} U^{t}_j&i\ge m^t,
\end{cases}
\end{equation}
where $m^t\in M^t$, as defined in \eqref{def:Mt}, and $\varphi^t_{\text{max}}=\phi_{\text{max}}+c_{\text{max}} t$ can be calculated directly from \eqref{eq:varphi}. 

\begin{thm}\label{th:BS}
If $V^t>0$, the pair $\bar\Theta^t_i$ and $\bar\Theta^{t+1}_i$ given by \eqref{eq:overline Theta} and \eqref{eq:overline Theta +1} respectively, satisfy the ultradiscrete linear equations \eqref{eq:lin1}--\eqref{eq:lin4}. 
\end{thm}
\begin{proof}
For linear equations \eqref{eq:lin1} and \eqref{eq:lin2} we may use Theorem~\ref{thm:gen sol} to prove that they are satisfied by $\bar \Theta^t_i$ and $\bar \Theta^{t+1}_i$ whenever all of the indices $i-1$, $i$ and $i+1$ lie in the same interval $(-\infty,m^t]$ or $[m^t,\infty)$. It remains to prove \eqref{eq:lin1} and \eqref{eq:lin2} when $i=m^t$. This is done below. Equations \eqref{eq:lin3} and \eqref{eq:lin4} only involve indices $i$ and $i+1$ which both lie in either $(-\infty,m^t]$ or $[m^t,\infty)$ and so Theorem~\ref{thm:gen sol} deals with all cases.

Because of the choice $X_{m^t}^t=\omax$, satisfying \eqref{eq:lin1} and \eqref{eq:lin2} at $i=m^t$ requires that
\begin{align}
\max(U^{t-1}_{m^t}+ U^{t}_{m^t}, U^{t}_{m^t-1}+ U^{t}_{m^t})&=
\max(U^{t}_{m^t}+\ U^{t}_{m^t-1}, U^{t+1}_{m^t-1}+ U^{t}_{m^t-1})\nonumber\\
\label{eq:proof}
&=\max( U^{t}_{m^t}+ U^{t}_{m^t-1}-1,0)+\kmax.
\end{align}
First consider the case in which $V^t\le1$. Then certainly $U^t_{m^t-1}+ U^t_{m^t}\le1$ and for all $i,t$, $U^{t-1}_i= U^t_{i+1}$ (cf. Corollary \ref{cor:cons}) and the maximum $\kmax=\omax$ is attained at $m^t$ (Corollary \ref{cor:newcor}). Hence the requirement \eqref{eq:proof} becomes
\begin{align*}
\max(U^{t}_{m^t+1}+ U^{t}_{m^t}, U^{t}_{m^t-1}+ U^{t}_{m^t})=\max( U^{t}_{m^t}+ U^{t}_{m^t-1}, U^{t}_{m^t-2}+ U^{t}_{m^t-1})=\kmax,
\end{align*}
which is satisfied since $U^t_{m^t}+ U^t_{m^t-1}=\omax=\kmax$ is maximal.  

Finally we deal with the case $\tilde V^t>1$ in which $\omax>1$ and $\kmax=1$. We proved in Lemma~\ref{lem:m} (i) that if $V^t>1$ and $X^t_m$ is maximal then either
\begin{equation}\label{eq:U cond 1}
U^t_{m^t-1}+ U^t_{m^t}\ge1
\end{equation} 
or
\begin{equation}\label{eq:U cond 2}
U^t_{m^t-1}+ U^t_{m^t}<1\text{~ and ~} U^{t-1}_m+ U^t_{m^t}= U^t_{m^t-1}+ U^{t+1}_{m^t-1}=1.
\end{equation} 

If \eqref{eq:U cond 1} holds then \eqref{eq:proof} becomes 
\begin{align*}
\max( U^{t-1}_{m^t}, U^{t}_{m^t-1})= U^{t}_{m^t-1}\text{ ~and ~}
\max( U^{t}_{m^t}, U^{t+1}_{m^t-1})= U^{t}_{m^t},
\end{align*}
and so $ U^{t-1}_{m^t}\le U^{t}_{m^t-1}$ and $ U^{t}_{m^t}\ge U^{t+1}_{m^t-1}$. These are simple consequences of \eqref{eq:naive} when \eqref{eq:U cond 1} holds.

On the other hand, if \eqref{eq:U cond 2} holds then \eqref{eq:proof} becomes
\begin{align*}
\max(1, U^{t}_{m^t-1}+ U^{t}_{m^t})=\max( U^{t}_{m^t}+ U^{t}_{m^t-1},1)=1,
\end{align*}
which is identically satisfied. 
\end{proof}

Finally, using the explicit dependence of $\varphi^t_{\text{max}}$ on the split point $m^t$ given by \eqref{eq:varphi}, we obtain an alternative expression for the eigenfunction \eqref{eq:overline Theta}:
\begin{equation}\label{eq:bar Theta alt}
\bar \Theta^t_i=
\begin{cases}
\displaystyle\kmax(i-m^t)+ \sum_{j=i}^{m^t-1} U^{t}_j + \sum_{j< m^t} U^{t-1}_j&i\le m^t\\
\displaystyle\sum_{j<i} U^{t-1}_j&i\ge m^t,
\end{cases}
\end{equation}
noting that the two cases in the formula agree at $i=m^t$. 

\begin{prop}\label{prop: cluster}
Let $\bar\Theta^t_i(m_1)$ and $\bar\Theta^t_i(m_2)$ be given by \eqref{eq:bar Theta alt} for $m^t=m_1,m_2$ where $m_1<m_2\in M^t$. Then for all $i$,
\begin{equation}
\bar\Theta^t_i(m_1)\ge\bar\Theta^t_i(m_2).
\end{equation}
In particular,  
\begin{equation}
\bar\Theta^t_i(m_1)=\bar\Theta^t_i(m_2),
\end{equation}
if and only if $m_1$ and $m_2$ are in the same block $M^t_j\subset M^t$.
\end{prop}
\begin{proof}
From \eqref{eq:bar Theta alt},
\begin{equation}
\bar\Theta^t_i(m_2)-\bar\Theta^t_i(m_1)=
\begin{cases}
\displaystyle\sum_{j=m_1}^{m_2-1}( U^t_j+ U^{t-1}_j-\kmax)&i\le m_1\\
\displaystyle\sum_{j=i}^{m_2-1}( U^t_j+ U^{t-1}_j-\kmax)&m_1<i<m_2\\
0&i\ge m_2.
\end{cases}
\end{equation}
Then, using \eqref{eq:basic kmax}, $ U^t_j+ U^{t-1}_j\le\kmax$ for all $j$ and it follows immediately that $\bar\Theta^t_i(m_1)\ge\bar\Theta^t_i(m_2)$.

If $m_1,m_2$ belong to the same block then, by Lemma~\ref{lem:m} (ii), $ U^t_j+ U^{t-1}_j=\kmax$ for all $m_1-1\le j\le m_2-1$, and hence $m_1\le j\le m_2-1$, and so it follows that $\bar\Theta^t_i(m_1)=\bar\Theta^t_i(m_2)$.

If $m_1,m_2$ do not belong to the same block then, since $ X^t_{m_1}$ and $ X_{m_2}^t$ are separated maxima, there is some $m_1\le k<m_2$ for which $\Delta  X_k<0$. Then by Proposition~\ref{prop:X ineq}, $ U^t_k+ U^{t-1}_k<\kmax$ and so $\bar\Theta^t_i(m_1)>\bar\Theta^t_i(m_2)$ for some $i$. 

Hence $\bar\Theta^t_i(m_1)=\bar\Theta^t_i(m_2)$ if and only if $m_1$ and $m_2$ belong to the same block $M_j^t$.
\end{proof}
Note that this also implies that, within the same block $M_j^t$, the value of $\phi_{\max}$ calculated from \eqref{eq:phi} is the same regardless the value of $m^t$ one chooses. This follows immediately from the asymptotic behaviour of the eigenfunction $\bar\Theta_i^t$:
\begin{equation}
\bar\Theta^t_i\sim \sum_{j\in\mathbb{Z}} U_j^t +
\begin{cases}
\kmax(i-\phi_\text{max}-c_\text{max} t)&\text{if }i<\ell^t\\
0&\text{if } i>r^t,  
\end{cases}
\end{equation}
with $\ell^t$ the left-most boundary of the support of $U_i^t$ and $r^t$ the right-most boundary of the support of $U_i^{t-1}$. 

This asymptotic behaviour also implies that for such eigenfunctions, the ultradiscrete analogue $S\!E(\bar\Theta_i^t)$ of a square eigenfunction 
\begin{equation}\label{def: SE}
S\!E(\bar\Theta_i^t):=\bar\Theta_i^t+\bar\Theta_{i-1}^t+\kappa_\text{max}(c_\text{max} t + 1 - i)
\end{equation}
has asymptotic behaviour
\begin{equation}
S\!E(\bar\Theta_i^t) \sim 2\sum_{j\in\mathbb{Z}} U_j^t +
\begin{cases}
\kmax(i-2\phi_\text{max}-c_\text{max} t)&\text{if }i<\ell^t\\
\kappa_\text{max}(c_\text{max} t + 1 - i)&\text{if } i>r^t+1,  
\end{cases}
\end{equation}
meaning that it decresases to $-\infty$ at both limits $i\to\pm\infty$, taking a finite maximal value in between. This shows that the eigenfunction $\bar\Theta_i^t$ is a natural ultradiscrete counterpart to the bound state eigenfunctions for the discrete KdV Lax pair we discussed in Section \ref{sec:darboux} (cf. Remark \ref{rem:discreteBS}). In this sense, Proposition \ref{prop: cluster} is telling us that in the ultradiscrete case the spectrum for the linear system is not simple: we can have different eigenfunctions for the same value of $\omax$. This reflects the possibility of having several solitons with the same mass in any given solution to the ultradiscrete KdV equation.

Notice that the middle, non-asymptotic, part of the ultradiscrete squared eigenfunction $S\!E(\bar\Theta_i^t)$ has a relatively small extent, typically smaller than the size of the support of the state $U_i^t$, and that the asymptotic wave fronts in $S\!E(\bar\Theta_i^t)$, for $i\to\pm\infty$, both move in the positive $i$ direction at the same speed $c_{\text{max}}$. We have already seen that 
$c_{\max}=\max(1,\omax)$ is the maximal soliton speed in the state $U_i^t$. The fact that the wave front facing the positive $i$ direction evolves unperturbed implies that, asymptotically, nothing can overtake it and that, in fact, none of the constituent parts of $U_i^t$ can move towards $i\sim+\infty$ at a speed greater than $c_{\text{max}}$.  Hence, the $\omax$ solitons are indeed the fastest structures contained in $U_i^t$.

We will now show that a Darboux transformation using $\bar\Theta_i^t$ removes one of these fastest solitons from $U_i^t$.

\subsection{The undressing transformation}\label{sec:undressing}
In this section we consider the effect of a Darboux transformation defined in terms of $\bar\Theta^t_i$. We will see that it gives a soliton-removing/undressing transformation $U^t_i\to \widehat U^t_i$. We shall also describe the action of this Darboux transformation on the $T$-functions for the state $U_i^t$ we wish to undress.

As we saw in Section \ref{sec:dressing}, the effect of a dressing transformation on a given solution of udKdV is very complicated and it seems impossible to give a simple, explicit, expression for the dressed state. The best we could do was the asymptotic formula \eqref{eq:tilde U split}.  The effect of an undressing Darboux transformation that removes a soliton can be described much more precisely. 

The undressing eigenfunction $\bar\Theta^t_i$ given by \eqref{eq:bar Theta alt} has a split point which is the same for $t$ and $t+1$ and which is completely determined by the state $U_i^t$ to which the transformation will be applied.
The target state for the transformation (or `undressed' state) will be denoted by $\widehat U^t_i$. Then, from \eqref{eq:bar Theta alt} we have
\begin{equation}
  \bar\Theta^t_{i+1}-\bar\Theta^t_i=
  \begin{cases}
    \kmax- U^{t}_i&i<m^t\\
     U^{t-1}_i&i\ge m^t.
  \end{cases}
\end{equation}
Since the split point for $\bar\Theta_i^{t+1}$ can be taken to be the same as for $\bar\Theta_i^{t}$,  the result of the undressing Darboux transformation is remarkably simple:
\begin{equation}\label{eq:undress U}
  \widehat U^t_i=U^t_i+(\bar\Theta^t_{i+1}-\bar\Theta^t_i)-(\bar\Theta^{t+1}_{i+1}-\bar\Theta^{t+1}_i)=
  \begin{cases}
    U^{t+1}_i&i<m^t\\
    U^{t-1}_i&i\ge m^t,
  \end{cases}
\end{equation}
(to be compared with \eqref{eq:tilde U split} for the dressing transformation). This formula tells us, roughly speaking, that the undressed solution is obtained by removing the soliton near to $m^t$ and filling the gap by joining the update of the left part and the downdate of the right part. The target state in the undressing is therefore again of finite support. Moreover, it is easily verified that the total mass of the initial state $ U_i^t$ is indeed reduced by $\omax$ in the undressing:
\begin{equation}
\sum_{i\in\mathbb{Z}} \widehat U_i^t = \sum_{i<m^t}  U_i^{t+1} + \sum_{i\ge m^t}  U_i^{t-1}= \sum_{i\in\mathbb{Z}}  U_i^{t-1} -X_{m^t}^t = \sum_{i\in\mathbb{Z}}  U_i^{t}-\omax,
\end{equation}
with $X_i^t$ as defined in \eqref{eq:X}. It is worth emphasizing that since $\bar\Theta_i^t$ can only be constructed for $\omega=\omax$, the undressing transformation we just obtained can only remove a soliton that has maximal mass in $U_i^t$.

\begin{exa}\label{exa:undressing1}
As an example, let us calculate the eigenfunction $\bar\Theta_i^t$ and target state $\widehat U_i^t$ for the undressing transformation that reverses the dressing of Example \ref{exa:dressing}, in which a soliton with mass $\omega=2$ was added to the initial state $\dots,0, 0, 0,0,0,1,\frac12,0,1,0,0,0,0,0,0,\dots$ at $t=0$, for a phase constant $\phi=7$. We therefore start from the state $U_i^0\ : \dots,0,0,0,\frac12,1,0,\frac12,1,0,1,\frac12,0,0,0,0,\dots$ (the $\tilde U_i^t$ state of Example \ref{exa:dressing}) and first analyse its soliton content. This will tell us which solitons can be removed and what the appropriate split points are.
\[
\setlength{\arraycolsep}{0pt}
\renewcommand{\arraystretch}{1.5}
  \begin{array}{rCCCCCCCCCCCCCCCCCCCCCCCCCCCCCCCCCCCCCC}
U_i^0\ :&&0&0&0&\frac12&1&0&\frac12&1&0&1&\frac12&0&0&0&0\\
X^0_i\ :&&0&0&\tfrac12&\tfrac32&2&2&2&2&2&2&2&\tfrac32&\tfrac12&0&0\\
\end{array}
\]
We find that an $\omax=2$ soliton can be removed at split points $m^t=5, \cdots, 11$. As all these split points belong to the same block in $X_i^0$ the appropriate value for $\phi_{\max}$ is unique and can, for example, be obtained from \eqref{eq:phi} at $m^t=5$ and $t=0$:
$$\phi_{\max}=5 + \sum_{j\ge5} U_j^{-1} - \sum_{j<5} U_j^0 = 5 + \tfrac52 - \tfrac12 = 7.$$
We will see shortly that it is not a coincidence that this value exactly matches the value of the phase constant that was used in the dressing. The downdate $U_i^{-1}$, together with the update $U_i^1$ of the initial state, can also be used to obtain the undressed state $\widehat U_i^0$ by means of \eqref{eq:undress U}:
\[
\setlength{\arraycolsep}{0pt}
\renewcommand{\arraystretch}{1.5}
  \begin{array}{rCCCCCCCCCCCCCCCCCCCCCCCCCCCCCCCCCCCCCC}
U^{1}_i\ :&&0&0&0&0&[0&1&\tfrac12&0&1&0&\tfrac12&1&\tfrac12&0&0]\\

U^{-1}_i\ :&&[0&\tfrac12&1&\tfrac12]&\,0&1&\tfrac12&0&1&0&0&0&0&0&0\\
\widehat U_i^0\ :&&0&0&0&0&\,0&1&\tfrac12&0&1&0&0&0&0&0&0,
\end{array}
\]
for $m^t=5$, where the bracketed values indicate the parts of  $U_i^1$ and $U_i^{-1}$ that are not used in the construction of $\widehat U_i^0$. We see that the target state $\widehat U_i^0$ indeed matches the initial state of Example \ref{exa:dressing}.

It is also interesting to calculate the bound state eigenfunction $\bar\Theta_i^0$ used in the undressing and to compare it with the eigenfunction $\Theta_i^0 = \dots, ,,0,0,0,0,\tfrac12,\tfrac32,\tfrac32,2,3,3,4,5,6,7,8, \dots$ that was used in Example \ref{exa:dressing} to dress the state $\widehat U_i^0$.
\[
\setlength{\arraycolsep}{0pt}
\renewcommand{\arraystretch}{1.5}
  \begin{array}{rCCCCCCCCCCCCCCCCCCCCCCCCCCCCCCCCCCCCCC}
(i-7) + \text{$\displaystyle\sum_{j\ge i} U_j^0$}\ :&&\overline{\tfrac32}&\overline{\tfrac12}&\tfrac12&\tfrac32&2&2&3&\tfrac72&\tfrac72&\tfrac92&\tfrac92&5&6&7&8\\
\text{$\displaystyle\sum_{j< i} U_j^{-1}$}\ :&&0&0&\tfrac12&\tfrac32&2&2&3&\tfrac72&\tfrac72&\tfrac92&\tfrac92&\tfrac92&\tfrac92&\tfrac92&\tfrac92\\
\bar\Theta_i^0\ :&&\overline{\tfrac32}&\overline{\tfrac12}&\tfrac12&\tfrac32&2&2&3&\tfrac72&\tfrac72&\tfrac92&\tfrac92&\tfrac92&\tfrac92&\tfrac92&\tfrac92\\
\text{$\displaystyle\sum_{j\in\mathbb{Z}} U_j^0+ (i-7) -\bar\Theta_i^0$}\ :&&0&0&0&0&\tfrac12&\tfrac32&\tfrac32&2&3&3&4&5&6&7&8,
\end{array}
\]
for $\sum_{j\in\mathbb{Z}} U_j^0=\frac{9}{2}$ and where  a bar designates negative values. We see that $\bar\Theta_i^0+\Theta_i^0=\sum_{j\in\mathbb{Z}} U_j^0+ (i-7)$ which suggests that the (generic) eigenfunction used in the dressing of Example \ref{exa:dressing} is in fact adjoint to $\bar\Theta_i^0$, in the sense of \eqref{eq:u 1/theta} (up to a renormalisation by $\sum_{j\in\mathbb{Z}} U_j^0 - \phi_{\max}$). We will see in the next section  that this is indeed the case and that, in general, there always exists a generic eigenfunction that reverses the effect of an undressing Darboux transformation.
\end{exa}

First, let us describe the effect of the undressing transformation on the $T$-function for $U_i^t$. Starting from the form \eqref{eq:bar Theta} for $\bar\Theta_i^t$ and introducing a discrete potential $Z_i^t=T^t_i-T^{t+1}_i$ for $U_i^t$ with boundary condition 
$Z_-=-\frac12\sum_{j\in\mathbb{Z}}U_j^t$ as $i\to-\infty$ (just as we did in Section \ref{sec:dressing}), we find $\sum_{j<i}U^{t}_j= \frac12\sum_{j\in\mathbb{Z}}U_j^t + T^{t}_i-T^{t+1}_i$, for all $t$ and we can write the undressed $T$-function, $\widehat T_i^t = T_i^t + \bar\Theta_i^t $, as 
\begin{equation}
\widehat T_i^t = \frac12\sum_{j\in\mathbb{Z}}U_j^t +\tfrac12\kappa_{\max} (i -\varphi_{\max}^t)  + \min\Big( \tfrac12 \kappa_{\max} (i -\varphi_{\max}^t) + T_i^{t+1} , 
-\tfrac12 \kappa_{\max} (i -\varphi_{\max}^t) + T_i^{t-1}\Big).
\end{equation}
Using gauge freedom, as before, this suggests that the transformation $T_i^t \mapsto \widehat T_i^t$,
\begin{equation}\label{Tundressing}
\widehat T_i^t \sim  \min\Big( \tfrac12 \kappa_{\max} (i -\varphi_{\max}^t) + T_i^{t+1} , 
-\tfrac12 \kappa_{\max} (i -\varphi_{\max}^t) + T_i^{t-1}\Big),
\end{equation}
and the dressing given by \eqref{eq:nakata} might be inverse to each other. More precisely, consider the following chain 
\begin{equation}\label{Tchain}
T_i^t \xrightarrow{~~\bar\Theta_i^t~} \widehat T_i^t \xrightarrow{~~\widehat\Theta_i^t~} \tilde T_i^t ,
\end{equation}
of an undressing Darboux transformation using a bound state $\bar\Theta_i^t$ with $(\omax, \phi_{\max})$  determined on the initial potential $U_i^t=T_{i+1}^t+T_i^{t+1}-T_i^t-T_{i+1}^{t+1}$, followed by a dressing Darboux transformation in terms of a generic eigenfunction $\widehat\Theta_i^t$ for the undressed potential $\widehat U_i^t = \widehat T_{i+1}^t + \widehat T_i^{t+1} - \widehat T_i^t - \widehat T_{i+1}^{t+1}$, with $\omega=\omax$ and $\phi=\phi_{\max}$. The resulting $T$-function is obtained, up to a gauge, by applying \eqref{eq:nakata} to $\widehat T_i^t$ (as in\eqref{Tundressing}) which gives:
\begin{equation}
\tilde T_i^t = T_i^t + \frac{\omax}{2} + {\cal M}(T),
\end{equation}
where
\begin{equation}
{\cal M}(T) =\max\Big(\min\big(0, \kappa_{\max} (i-\varphi_{\max}^{t+1}) + T_i^{t+2}-T_i^t\big) ,\min\big(0, -\kappa_{\max} (i-\varphi_{\max}^{t-1}) + T_i^{t-2}-T_i^t\big)\Big).
\end{equation}
Asymptotically, it is easy to check that $\lim_{i\to\pm\infty} (T_i^{t+2}-T_i^t) = -2Z_\pm$ and $\lim_{i\to\pm\infty} (T_i^{t-2}-T_i^t) = 2Z_\pm$ where, as before, the constants $Z_+$ and $Z_-$ are the asymptotic values of the discrete potential $Z_i^t$ at $i\to+\infty$ and $i\to-\infty$ respectively. Hence ${\cal M}(T)=0$ for all $t$, for large enough $|i|$. It seems difficult to show, in all generality, that ${\cal M}(T)\equiv 0$ for all $i$ and $t$ but it was brought to our attention \cite{nakata-priv} that such a result can be proven for the reverse situation where \eqref{Tundressing} is used to remove a soliton that was first added by means of \eqref{eq:nakata}. There is however a different way to prove the desired result.

\subsection{Reversing an undressing transformation by dressing}\label{sec:inverse}
Consider the equivalent of the undressing/dressing chain \eqref{Tchain} on $U_i^t$:
\begin{equation}\label{Uchain}
U_i^t \xrightarrow{~~\bar\Theta_i^t~} \widehat U_i^t \xrightarrow{~~\widehat\Theta_i^t~} \tilde U_i^t.
\end{equation}
As shown in \eqref{eq:undress U}, the effect of undressing $U_i^t$ by a Darboux transformation \eqref{eq:ud darboux} in terms of a $\bar\Theta_i^t$ given by \eqref{eq:overline Theta}, is simply
\begin{equation}\label{UtoUhat}
 U_i^t \mapsto~ \widehat U^t_i=
  \begin{cases}
    U^{t+1}_i&i<m^t\\
    U^{t-1}_i&i\ge m^t,
  \end{cases}
\end{equation}
for a split point $m^t\in M^t$ \eqref{def:Mt} determined on $U_i^t$. As this undressing has taken out a soliton with mass $\omax$, the remaining solitons in $U_i^t$ necessarily have masses less than or equal to $\omax$ and we can therefore always consider a dressing Darboux transformation on $\widehat U_i^t$ using a generic eigenfunction 
$\widehat\Theta_i^t$ for $\widehat U_i^t$ with $\omega=\omax$. Since the phase constant $\phi$ in this generic eigenfunction is completely free, we can always choose it to coincide with $\phi_{\max}$ of $\bar\Theta_i^t$, given by \eqref{eq:phi}. {Moreover, as will see in Lemma \ref{lem:samesplit}, there is always an index $m\in\mathbb{Z}$ that conforms to both the definition of a split point for the bound state eigenfunction $\bar\Theta_i^t$ (for $U_i^t$) as well as to that for a split point for the generic eigenfunction 
$\widehat\Theta_i^t$ (for $\widehat U_i^t$) we are considering here. Taking this index as the common split point $m^t$ for both functions, we write } $\widehat\Theta_i^t$ as 
\begin{equation}\label{def:thetahat}
  \widehat\Theta^t_i=
  \begin{cases}
    \displaystyle\sum_{j<i} \widehat U^{t-1}_j&i<m^t\\
    \displaystyle\kappa_{\max} (i-\varphi_{\max}^t)+\sum_{j\ge i} \widehat U^t_j&i\ge m^t
  \end{cases}
\end{equation}
(cf. \eqref{eq:max sol alt}) where the split point $m^t$ is now exactly the same as that in \eqref{UtoUhat}. Taken together then, these two relations tell us that
\begin{equation}
  \widehat\Theta^t_i=
  \begin{cases}
    \displaystyle\sum_{j<i} U^{t}_j&i<m^t\\
    \displaystyle\kappa_{\max} (i-\varphi_{\max}^t)+\sum_{j\ge i} U^{t-1}_j&i\ge m^t,
  \end{cases}
\end{equation}
or, because of \eqref{eq:overline Theta}, that we have
\begin{equation}\label{eq:adjointrel}
\bar\Theta_i^t + \widehat\Theta_i^t = \kappa_{\max} (i-\varphi_{\max}^t)+\sum_{j\in\mathbb{Z}} U^{t}_j ,\quad \text{for all~}  i, t\in\mathbb{Z}.
\end{equation}
Since both $\bar\Theta_i^t$ and $\widehat\Theta_i^t$ solve the linear system \eqref{eq:lin1}--\eqref{eq:lin4}, for the respective potentials $U_i^t$ and $\widehat U_i^t$ which are related by Darboux transformation, we conclude from \eqref{eq:adjointrel} that $\bar\Theta_i^t$ and $\widehat\Theta_i^t$ are in fact related as in \eqref{eq:u 1/theta}, up to a trivial renormalisation by $\sum_{j\in\mathbb{Z}} U^{t}_j - \kappa_{\max} \phi_{\max}$. Hence, 
\begin{equation}
U_i^t \xrightarrow{~~\bar\Theta_i^t~} \widehat U_i^t \xrightarrow{~~\widehat\Theta_i^t~} \tilde U_i^t \equiv U_i^t,
\end{equation}
and we find that:
\begin{thm}\label{inversenature}
For any undressing Darboux transformation,  defined by a bound state eigenfunction $\bar\Theta_i^t$ \eqref{eq:overline Theta}, there exists a dressing Darboux transformation that reinserts the soliton that was taken out in the undressing, thereby reconstructing the original solution to udKdV one had before undressing. 
\end{thm}
\begin{proof}
As is clear from the above arguments, such a dressing transformation is obtained from a (generic) eigenfunction $\widehat\Theta_i^t$ as in \eqref{def:thetahat} for the undressed state $\widehat U_i^t$,   if we can take the split point $m^t$ for $\widehat\Theta_i^t$ such that it coincides with a valid split point for the bound state eigenfunction $\bar\Theta_i^t$, in the block $M_j^t$ that corresponds to the soliton that was taken out from the original solution. This is always possible by Lemma \ref{lem:samesplit}.
\end{proof}

The following Lemma will be shown in Appendix \ref{proofLemmasamesplit}.
\begin{lem}\label{lem:samesplit}
If we consider an undressing of a state $U_i^t$ by a bound state eigenfunction  $\bar\Theta_i^t$ \eqref{eq:overline Theta} with $(\omax,\phi_{\max})$, resulting in an undressed state $\widehat U_i^t$, then the left-most index in the block $M_j^t$ in $X_i^t$ \eqref{eq:X} that corresponds to the soliton that was taken out by the undressing, is also a valid split point for a generic eigenfunction $\widehat\Theta_i^t$ with  $(\omax,\phi_{\max})$  \eqref{def:thetahat} for $\widehat U_i^t$.
\end{lem}

Theorem \ref{inversenature}, in fact,  offers a practical method for solving the Cauchy problem for the udKdV equation.
\begin{exa}\label{exa:undressing2} 
Let us look at the undressing and subsequent dressing of the undressed state $\widehat U_i^0$ obtained in Example \ref{exa:undressing1} (which, for simplicity however, we shall denote here as $U_i^t$). For this state we find
\[
\setlength{\arraycolsep}{0pt}
\renewcommand{\arraystretch}{1.5}
  \begin{array}{rCCCCCCCCCCCCCCCCCCCCCCCCCCCCCCCCCCCCCC}
U_i^0\ :&&0&0&0&0&\,0&1&\tfrac12&0&1&0&0&0&0&0&0\\
X^0_i\ :&&0&0&0&0&\tfrac12&\tfrac32&\tfrac32&1&1&1&0&0&0&0&0,
\end{array}
\]
which shows that it contains a soliton with mass $\omax=\tfrac32$ that can be taken out at split points 6 or 7. The phase constant $\phi_{\max}$ can be calculated as (at $m^t=6$ and $t=0$)
$$\phi_{\max}=6 + \sum_{j\ge6} U_j^{-1} - \sum_{j<6} U_j^0 = 6 + 1 - 0 = 7,$$
and the undressed state $\widehat U_i^0$ is found to be:
\[
\setlength{\arraycolsep}{0pt}
\renewcommand{\arraystretch}{1.5}
  \begin{array}{rCCCCCCCCCCCCCCCCCCCCCCCCCCCCCCCCCCCCCC}
U^{1}_i\ :&&\,0&0&0&0&0&[0&\tfrac12&1&0&1&\tfrac12&1&\tfrac12&0&0]\\
U^{-1}_i\ :&&[0&\,0&0&\tfrac12&1]&\,0&0&1&0&0&0&0&0&0&0\\
\widehat U_i^0\ :&&\,0&0&0&0&0&\,0&0&1&0&0&0&0&0&0&0,
\end{array}
\]
where the bracketed values again indicate the parts of  $U_i^1$ and $U_i^{-1}$ that are not used in the construction of $\widehat U_i^0$ (for $m^t=6$). In this case, the undressed state is very simple and clearly moves at speed 1 towards the right. Hence, we can immediately write its $T$-function, as if it were part of a background (cf. formula \eqref{backgroundT})
\begin{equation}
\widehat T^t_i=\dfrac{1}{2} |i-t-8|.
\end{equation}
If we now re-insert the mass $\tfrac32$ soliton using \eqref{eq:nakata}, for $\omega=\tfrac32$ and $\phi=7$, we obtain the $T$-function
\begin{equation}\label{analyticT}
T_i^t=\frac12 \max\Big(i-7 - \tfrac32 t+ |i-t-9|,- (i-7 - \tfrac32 t)+ |i-t-7|\Big),
\end{equation}
from which we can calculate an explicit function $U_i^t=T_{i+1}^t + T_{i}^{t+1} - T_i^t - T_{i+1}^{t+1}$ that coincides with $U_i^0$ at $t=0$.
In other words, we have in fact solved the Cauchy problem for the udKdV equation with initial condition $U_i^0$. Figure \ref{fig:Ex7} shows the time evolution of the initial state $U_i^0$ calculated using the update rule \eqref{eq:update rule}, plotted as dots, together with the exact solution $U_i^t$ found from \eqref{analyticT} with $i$ taken to be a real variable, plotted as a continuous line. We observe that, as they should, the two plots coincide at integer values of $i$.
\begin{figure}[htbp]\centering
\includegraphics*[width=12.5cm]{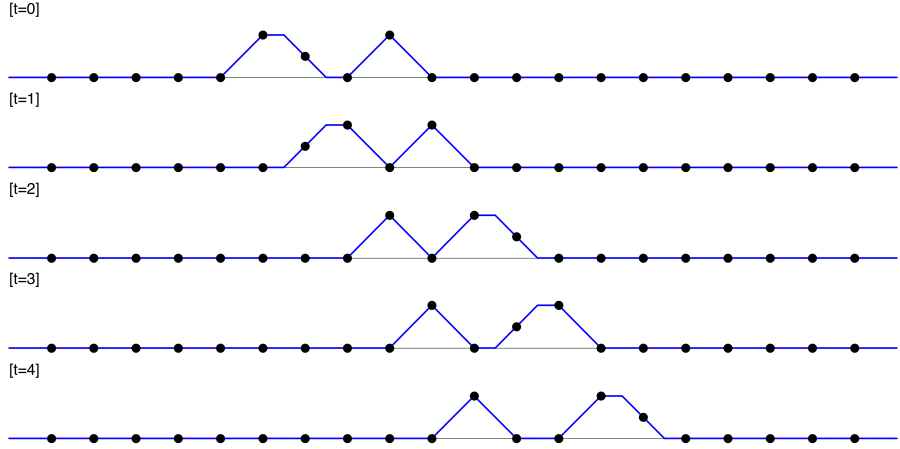}\vskip-.2cm
\caption{\label{fig:Ex7} Simulation (dots) and exact solution obtained from \eqref{analyticT} (solid line) shown together.}
\end{figure}

\end{exa}

\section{Solution to the Cauchy problem for udKdV}\label{sec:cauchy}
We have seen that the dressing and undressing transformations, using the same parameters $\omega$ and $\phi$, are inverse to one another. As described in \cite{MR2738130}, this leads to a method for solving the Cauchy problem for udKdV. This method can be summarized as:

\begin{thm}\label{thm:cauchy}
{For the udKdV equation \eqref{eq:update rule} any given initial state $U_i^0$, with finite support}, can be completely undressed until only a background state remains. The data $(\omega, \phi)$ obtained at each step in the undressing, together with the background state, suffice to construct an exact solution $U_i^t$ for  \eqref{eq:update rule} that coincides with  $U_i^0$ at $t=0$.
\end{thm}
\begin{proof}
That any initial state  $U_i^0$ with finite support can be fully undressed, down to a pure background state (i.e. a state for which $V^t=0$) follows immediately from the fact that an undressing transformation reduces the extent of the state (cf. \eqref{eq:undress U}), which means that  after repeated undressings one eventually ends up with the trivial state, or with a state for which $\omax=0$ and further undressing is impossible. In either case one has $V^t=0$ (see Remark \ref{newposremark5}).
We can then proceed as follows:
\begin{enumerate}
  \item Given initial data $U^0_i$, use an undressing transformation to remove one of the heaviest solitons and record $\omega=\omax$ and $\phi=\phi_{\max}$ for this soliton. 
  \item Repeat until all solitons are removed and only a background state $B^0_i$ remains, at which point we have obtained the full set of spectral data: a finite number of pairs $(\omega_j, \phi_j)$ ($j=1, \dots, N$) for non-increasing soliton masses ($\omega_1\geq \omega_2 \geq \cdots \geq \omega_N$) and the background $B_i^0$, or just a background if the initial state did not contain any solitons.
  \item Evolve the background---it simply translates at speed 1---to give $B^t_i=B^0_{i-t}$. 
  
 This is straightforward if calculated in terms of the $T$-function \eqref{backgroundT} for the background, which we shall denote by $T_i^{[0],t}$.
  \item Add back all the solitons in reverse order using the sets of parameters $(\omega, \phi)$ obtained at each undressing, to obtain the exact solution as a function of $i$ and $t$. 
  
This can be done by simple iteration of the dressing transformation \eqref{eq:nakata} with parameters $(\omega_{N-j},\phi_{N-j})$, 
\begin{equation}
T_i^{[j],t} \xrightarrow{~(\omega_{N-j},\phi_{N-j})\,} T_i^{[j+1],t},
\end{equation}
for $j$ running from $0$ to $N-1$. Since dressing transformations for the same $\omega$ commute \cite{MR2545618}, we always obtain a unique ``fully dressed" $T$-function, $T_i^{[N],t}$, which is guaranteed to solve the bilinear udKdV equation \eqref{eq:H udKdV}. From this $T$-function we can then calculate $U_i^t=T_{i+1}^{[N],t}+T_i^{[N],t+1}-T_i^{[N],t}-T_{i+1}^{[N],t+1}$, which solves the udKdV equation \eqref{eq:update rule} and which, by construction, coincides with $U_i^0$ at $t=0$.
\end{enumerate}\vskip-.2cm
\end{proof}

\begin{rem}\label{rem:separation}
Notice that the first statement in Theorem \ref{thm:cauchy} tells us that any initial state for which $V^t>1$ must, asymptotically, separate into a train of solitons with speeds greater than 1 (see also Remark \ref{rem:omega max}) and a remaining part that travels, unchanged, with speed 1 and that consists of solitons with $\omega\leq 1$, possibly embedded into a background.
\end{rem}

\begin{rem}
When solving the Cauchy problem according to the above algorithm, we can halt the undressing part as soon as we obtain a state  for which $\omax=1$,  because such a state simply evolves unchanged at speed 1 (Lemma \ref{lem:V^t}) and we can obtain its $T$-function by means of \eqref{backgroundT} as if it were a pure background state.
\end{rem}

\begin{cor}\label{cor:darbouxaction}
The undressing transformation \eqref{eq:undress U}, in terms of a bound state eigenfunction \eqref{eq:overline Theta}, maps between (finite support) solutions of the udKdV equation and the corresponding undressing transformation for $T$-functions \eqref{Tundressing} maps between solutions of the bilinear udKdV equation \eqref{eq:H udKdV}.
\end{cor}
\begin{proof}
Suppose that $U_i^t$ is a solution to the udKdV equation (of finite support) that contains at least one soliton. For this solution, let $T_i^{[N-1],t}$ denote the $T$-function obtained after $N-1$ dressings in the algorithm in the proof of Theorem \ref{thm:cauchy}, for some particular choice of initial undressing $(\omega_1, \phi_1)$, and let $T_i^t$ denote the result of the last dressing, i.e.: $U_i^t=T_{i+1}^{t}+T_i^{t+1}-T_i^{t}-T_{i+1}^{t+1}$ and
\begin{equation}
T_i^t = T_i^{[N-1],t} + \widehat  \Theta_i^t,
\end{equation}
where $\widehat \Theta_i^t$ is the generic eigenfunction \eqref{def:thetahat} for $\widehat U_i^t=T_{i+1}^{[N-1],t}+T_i^{[N-1],t+1}-T_i^{[N-1],t}-T_{i+1}^{[N-1],t+1}$ with $\omax=\omega_1$ and $\phi_{\max}=\phi_1$. Note that, by construction, $T_i^{[N-1],t}$ satisfies the bilinear udKdV equation \eqref{eq:H udKdV}. Then, because of \eqref{eq:adjointrel}, we have that the bound state eigenfunction $\bar \Theta_i^t$ for $(\omega_1, \phi_1)$ that undresses $T_i^t$ to $\widehat T_i^t=T_i^t + \bar \Theta_i^t$ (and therefore $U_i^t$ to $\widehat U_i^t = \widehat T_{i+1}^t + \widehat T_i^{t+1} - \widehat T_i^t - \widehat T_{i+1}^{t+1}$) can be expressed as 
\begin{equation}
\bar\Theta_i^t = \kappa (i-\varphi^t)+\sum_{j\in\mathbb{Z}} U^{t}_j - \widehat\Theta_i^t ,
\end{equation}
and that $\widehat T_i^t$ and $T_i^{[N-1],t}$ are therefore gauge equivalent:
\begin{equation}
\widehat T_i^t = T_i^{[N-1],t} + \kappa (i-\varphi^t)+\sum_{j\in\mathbb{Z}} U^{t}_j.
\end{equation}
Hence, the undressed $T$-function $\widehat T_i^t$ solves the bilinear equation  \eqref{eq:H udKdV} and $\widehat U_i^t$, which has finite support,  therefore solves the udKdV equation.

\end{proof}

To finish we give a detailed worked example of the undressing and redressing procedure.
\begin{exa}
Consider the initial state 
\begin{equation}
U^0_i=(\dots,0,0,0,\tfrac23,\tfrac23,-\tfrac12,1,\tfrac12,1,1,0,0,-\tfrac13,1,1,1,-1,1,1,1,0,0,0,0,0,0\dots).
\end{equation}
First we will show how to characterise the soliton content in terms of the {spectral data, i.e.:} pairs $(\omega,\phi)$ and a background. Then we reconstruct the solution from this data at an arbitrary time. As above, we use a bar notation for negative numbers. First we determine the data for the maximal soliton(s) in $U^0_i$ ; {the left-most 0 displayed has index $i=1$. Recall that  $X^0_i=\sum_{j<i}(U^{-1}_j-U^{1}_j)$, as defined in \eqref{eq:X}.}
\[
\setlength{\arraycolsep}{1pt}
\renewcommand{\arraystretch}{1.5}
\begin{array}{rCCCCCCCCCCCCCCCCCCCCCCCCCCCCCCCCCCCCCCCCCCCCC}
U^0_i\ :&&0& 0& 0& \tfrac23& \tfrac23& \overline{\tfrac12}& 1&  \tfrac12& 1& 1& 0& 0& \overline{\tfrac13}& 1& 1& 1&\overline1& 1& 1& 1& 0& 0& 0& 0& 0& 0\\
U^{1}_i\ :&&0& 0& 0& 0&\tfrac13& 1&   \overline{\tfrac12} & \tfrac12 & 0& 0& 1& 1& 1&  \overline{\tfrac13}& 0& 0& 2& 0& 0& 0& 1& 1& 1& 0& 0& 0\\
U^{-1}_i\ :&&0& \tfrac23 & 1& \tfrac13&\tfrac13&  \tfrac32& 0& \tfrac12 & 0& 0& \tfrac13& 1&  \tfrac43& 0& 0& 0& 2& 0& 0& 0& 0& 0& 0& 0& 0& 0\\
U^{-1}_i-U^{1}_i\ :&&0& \tfrac23 & 1& \tfrac13& 0& \tfrac12 & \tfrac12 & 0& 0& 0& \overline{\tfrac23}& 0& \tfrac13&\tfrac13& 0& 0& 0& 0& 0& 0& \overline1& \overline1& \overline1& 0& 0& 0\\
X^0_i\ :&&0& 0& \tfrac23 &  \tfrac53& 2& 2&  \tfrac52& 3& 3& 3& 3& \tfrac73&\tfrac73& \tfrac83& 3& 3& 3& 3& 3& 3& 3& 2& 1& 0& 0& 0,
\end{array}
\]
{from which it is clear that $X_i^0$} attains its maximum $\omax=3$ in two disjoint clusters, at $m^0\in\{8,9,10,11\}$ or $m^0\in\{15,16,17,18,19,20,21\}$. {We can remove a soliton from the state at $t=0$ by using the undressing \eqref{eq:undress U} with $m^0$ taking any one of these maximising values.} The phase parameter for this soliton is given by the formula \eqref{eq:phi}, at $t=0$,
\begin{equation}
\phi=m^0+\frac1{\kappa_{\text{max}}}\Big(\sum_{j\ge m^0}U^{-1}_j-\sum_{j<m^0}U^{0}_j\Big).
\end{equation}

By Proposition~\ref{prop: cluster}, {the result of the undressing only depends} on which cluster $m^0$ belongs to and so, for example, the choices $m^0=8$, 9, 10,11 all give the same result. We first choose $m^0=11$, {removing} the left-most $3$-soliton
\[
\setlength{\arraycolsep}{1pt}
\renewcommand{\arraystretch}{1.5}
\begin{array}{rCCCCCCCCCCCCCCCCCCCCCCCCCCCp{1.5in}}
U^{0}_i\ (i<11):     &&0& 0& 0& \tfrac23& \tfrac23& \overline{\tfrac12}& 1&  \tfrac12& 1& 1&&& && &&&&& &&&&&&&($\text{sum}=\tfrac{13}3$)\\
U^{1}_i\ (i<11):     &&0& 0& 0& 0&\tfrac13& 1&   \overline{\tfrac12} & \tfrac12 & 0& 0& & &&  &&&&&&&&&&&&\\
U^{-1}_i\  (i\ge11):&&& &&&& &  & &&& \tfrac13& 1&  \tfrac43& 0& 0& 0& 2& 0& 0& 0& 0& 0& 0& 0& 0& 0&($\text{sum}=\tfrac{14}3$)\\
\text{undressed }U^0_i\ : &&0& 0& 0& 0&\tfrac13& 1&   \overline{\tfrac12} & \tfrac12 & 0& 0& \tfrac13& 1&  \tfrac43& 0& 0& 0& 2& 0& 0& 0& 0& 0& 0& 0& 0& 0
\end{array}
\]
In this case, $\kappa_{\text{max}}=\min(1,\omax=3)=1$ and so 
\begin{equation}
  \phi=11+\tfrac{14}3-\tfrac{13}3=\tfrac{34}3.
\end{equation}
{The spatial extent of the undressed state (which we denote $\widetilde U_i^0$) is clearly smaller than that of the initial state $U_i^0$ and its mass has been reduced by 3. Furthermore, analysing the solitonic content of $\widetilde U_i^0$,  
\[
\setlength{\arraycolsep}{1pt}
\renewcommand{\arraystretch}{1.5}
\begin{array}{rCCCCCCCCCCCCCCCCCCCCCCCCCCCCCCCCCCCCCCCCCCCCC}
\widetilde U^0_i\ :&&0& 0& 0& 0&\tfrac13& 1&   \overline{\tfrac12} & \tfrac12 & 0& 0& \tfrac13& 1&  \tfrac43& 0& 0& 0& 2& 0& 0& 0& 0& 0& 0& 0& 0& 0\\
\widetilde U^{1}_i\ :&&0&0& 0& 0& 0&0& \tfrac43&   \overline{\tfrac12} & \tfrac12 & 0& 0& 0& \overline{\tfrac13}& 1& 1& 1& \overline{1}& 1& 1& 1& 0& 0& 0& 0& 0& 0\\
\widetilde U^{-1}_i\ :&&0&0&0& \tfrac23 & \tfrac23&\overline{\tfrac12}&  \tfrac12& \tfrac13&1 & 1& \tfrac23& 0& \overline{\tfrac13}& 1& 1& 1& \overline{1}& 0& 0& 0& 0& 0& 0& 0&0&0\\
\widetilde U^{-1}_i- \widetilde U^{1}_i\ :&&0&0&0& \tfrac23 & \tfrac23&\overline{\tfrac12}&  \overline{\tfrac56}& {\tfrac56}&\tfrac12 & 1& \tfrac23& 0&0& 0&0& 0& 0&\overline{1}&\overline{1}& \overline{1}& 0& 0& 0& 0&0&0\\
\widetilde X^0_i\ :&&0& 0&0&0& \tfrac23 &  \tfrac43& \tfrac56& 0&  \tfrac56& \tfrac43& \tfrac73& 3& 3& 3&3& 3& 3& 3& 2& 1& 0& 0& 0& 0& 0& 0,
\end{array}
\]
we see that the left-most 3-soliton has indeed disappeared.
}

Alternatively, we could choose $m^0=15,16,17,18,19,20,21$ {(which all give the same result) to remove a $3$-soliton on the right}.  We will take $m^0=17$
\[
\setlength{\arraycolsep}{1pt}
\renewcommand{\arraystretch}{1.5}
\begin{array}{rCCCCCCCCCCCCCCCCCCCCCCCCCCCp{1.5in}}
U^0_i\ (i<17) : &&0& 0& 0& \tfrac23& \tfrac23& \overline{\tfrac12}& 1&  \tfrac12& 1& 1&0&0&\overline{\tfrac13} &1&1 &1&&&&&&&&&&&($\text{sum}=7$)\\
U^{1}_i\ (i<17) :  &&0& 0& 0& 0&\tfrac13& 1&   \overline{\tfrac12} & \tfrac12 & 0& 0&1&1&1&\overline{\tfrac13}&0&0&\\
U^{-1}_i\  (i\ge17) :&&&&&&&&&&&&&&&&&&2&0&0&0&0&0&0&0&0&0&($\text{sum}=2$)\\
\text{undressed }U^0_i\ :&&0& 0& 0& 0&\tfrac13& 1&   \overline{\tfrac12} & \tfrac12 & 0& 0&1&1&1&\overline{\tfrac13}&0&0&2&0&0&0&0&0&0&0&0&0,
\end{array}
\]
and obtain
\begin{equation}
  \phi=17+2-7=12,
\end{equation}
{as the phase constant for the 3-soliton we removed.}

  If we continue with the  $m^0=17$ case and again calculate  $X^0_i$  for the {undressed state $\widetilde U_i^0$} we get:
\[
\setlength{\arraycolsep}{1pt}
\renewcommand{\arraystretch}{1.5}
\begin{array}{rCCCCCCCCCCCCCCCCCCCCCCCCCCCCCCCCCCCCCCCCCCCCCCCCCCC}
\widetilde U^0_i\ :&&0&0&0&0&\tfrac13& 1& \overline{\tfrac12}& \tfrac12& 0& 0& 1& 1& 1& \overline{\tfrac13}& 0& 0& 2&0&0&0&0&0&0&0&0&0\\
\widetilde U^{1}_i\ :&&0&0&0& 0& 0& 0& \tfrac43& \overline{\tfrac12}& \tfrac12& 0& 0& 0& 0& \tfrac43& 1& \tfrac13& \overline{1}& 1&1& 1& 0&0&0&0&0&0\\
\widetilde U^{-1}_i\ :&&0&0&0&\tfrac{2}{3}&\tfrac{2}{3}&\overline{\tfrac{1}{2}}&1&\tfrac{1}{2}&1&1&0&0&\overline{\tfrac{1}{3}}&1&1&1&\overline1&0&0&0&0&0&0&0&0&0\\
\widetilde U^{-1}_i-\widetilde U^{1}_i\ :&&0&0&0&\tfrac{2}{3}&\tfrac{2}{3}&\overline{\tfrac{1}{2}}&\overline{\tfrac{1}{3}}&1&\tfrac{1}{2}&1&0&0&\overline{\tfrac{1}{3}}&\overline{\tfrac{1}{3}}&0&\tfrac{2}{3}&0&\overline1&\overline1&\overline1&0&0&0&0&0&0\\
\widetilde X^0_i\ :&&0&0&0&0&\tfrac{2}{3}&\tfrac{4}{3}&\tfrac{5}{6}&\tfrac{1}{2}&\tfrac{3}{2}&2&3&3&3&\tfrac{8}{3}&\tfrac{7}{3}&\tfrac{7}{3}&3&3&2&1&0&0&0&0&0&0.
\end{array}
\]
{Notice here that although a 3-soliton has been removed on the right, there are still two disjoint clusters where $\widetilde X_i^0$ attains its maximum $\omax=3$.}   This tells us that the right-most cluster must, originally, have had a least two 3-solitons contained in it.   Indeed it turns out that there are three 3-solitons in the system altogether,  one in the left-most cluster of 3's and two in the right-most cluster of 3's {in $X_i^0$}.  {The order in which they are removed does not matter:  there are three different possibilities  LRR (left then right then right),  RLR or RRL but once all 3-solitons have been removed, the end result does not depend on the order of the undressings. Since the undressing transformation is valid for all times $t$,  this is easily checked on the asymptotic state $t \rightarrow +\infty$, where all 3-solitons are well separated and for which the effect of the undressing  \eqref{eq:undress U} is explicit.
}

We can repeat the undressing procedure until all solitons have been removed. {The intermediate states we obtain depend on the exact sequence of undressings we chose, but the final soliton-free state is independent of the order in which we perform the undressings.
 The complete process is summarised in the following table, in which we list besides the initial and final states, all intermediate states (for one particular order of the undressings), the split points $m^0$ at which a soliton is removed from each state as well as the value of $\omega$ and the phase constant for that soliton:}
\[
\setlength{\arraycolsep}{1pt}
\renewcommand{\arraystretch}{1.5}
\begin{array}{cccccccCCCCCCCCCCCCCCCCCCCCCCCCCCccccccccc}
&&&&\multicolumn{21}{c}{U^0_i}&&&&&&&m^0&&&&(\omax,\phi_{\text{max}})\\
\hline
{\rm (a)}&&&&0& 0& 0& \tfrac23& \tfrac23& \overline{\tfrac12}& 1&  \tfrac12& 1& 1& 0& 0& \overline{\tfrac13}& 1& 1& 1&\overline1& 1& 1& 1& 0& 0& 0& 0& 0& 0&&17&&&&(3,12)\\
{\rm (b)}&&&&0&0&0&0&\tfrac13& 1& \overline{\tfrac12}& \tfrac12& 0& 0& 1& 1& 1& \overline{\tfrac13}& 0& 0& 2&0&0&0&0&0&0&0&0&0&&11&&&&(3,\frac{34}3)\\
{\rm (c)}&&&&0&0&0& 0& 0& 0& \tfrac43& \overline{\tfrac12}& \tfrac12& 0&0&0&\overline{\tfrac{1}{3}}&1&1&1&\overline1&0&0&0&0&0&0&0&0&0&&15&&&&(3,12)\\
{\rm (d)}&&&&0&0&0&0&0&0&\overline{\tfrac13}&\frac32&\overline{\tfrac13}&\tfrac12&0&0&0&\overline{\tfrac13}&0&\overline{1}&0&0&0&0&0&0&0&0&0&0&&8&&&&(\frac43,7)\\
{\rm (e)}&&&&0&0&0&0&0&0&0&\overline{\tfrac12}&\frac12&0&0&0&\overline{\tfrac13}&0&\overline{1}&0&0&0&0&0&0&0&0&0&0&0&&10&&&&(\tfrac12, \tfrac{22}3)\\
{\rm (f)}&&&&0&0&0&0&0&0&0&0&\overline{\tfrac12}&0&0&\overline{\tfrac13}&0&\overline{1}&0&0&0&0&0&0&0&0&0&0&0&0&&\multicolumn{5}{c}{\text{(no soliton).}}
\end{array}
\]
In Figure~\ref{fig:undress} we illlustrate {how the values of $\omax$ and $m^0$ are determined for each state ((a) $\sim$ (f)). }  

\begin{figure}[htbp]\centering
\includegraphics*[width=11.5cm]{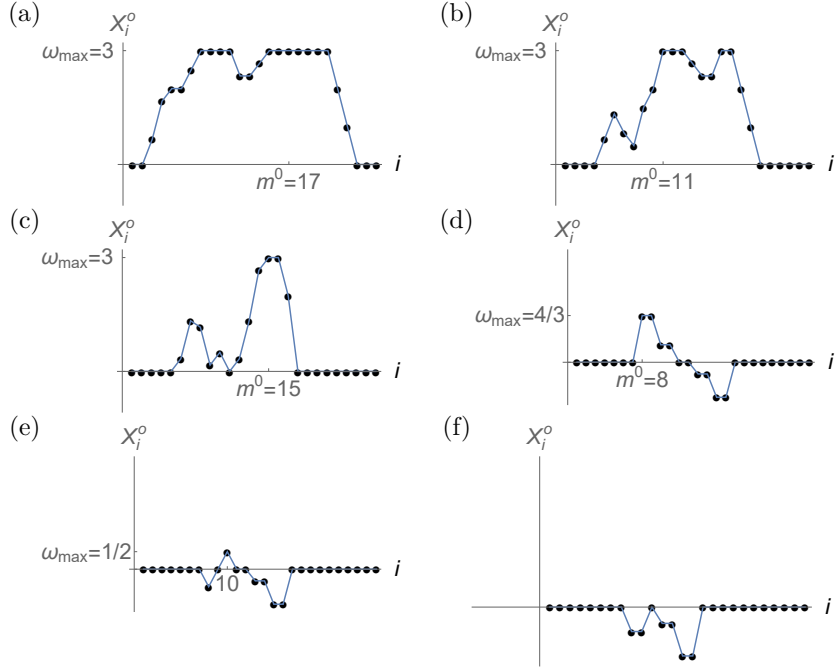}\vskip-.15cm
\caption{\label{fig:undress} Finding $\omega_{\text{max}}$ and $m^0$ after successive undressings of the initial data.}
\end{figure}

The solution at a general time $t$ may be reconstructed using the dressing transformation \eqref{eq:nakata}. To begin, we write the background state at time $t$, $B^t_i$, in terms of $T$-functions. 
{As explained in Section \ref{sec:sol prop}, the background $T$-function is given by}
\begin{equation}
  T^t_i=\tfrac12 \sum_{j\in\mathbb Z}|i-t-j|B^0_j .
\end{equation}
For the current example,
\begin{equation}
\setlength{\arraycolsep}{1pt}
\renewcommand{\arraystretch}{1.5}
\begin{array}{ccCCCCCCCCCCCCCCCCCCCCCc}
B^0_i&=&0&0&0&0&0&0&0&0&\overline{\tfrac12}&0&0&\overline{\tfrac13}&0&\overline{1}&0&0&0&0&0&0&0
\end{array}
\end{equation}
where the index of the first zero is 1, and so the background $T$-function is 
\begin{equation}
  T^t_i=-\tfrac14|i-t-9|-\tfrac16|i-t-12|-\tfrac12|i-t-14|
\end{equation}

Next, the solitons are added back using the data $(\omega,\phi)$ collected above, in weakly increasing order for the mass $\omega$,  by iterating transformation \eqref{eq:nakata}
\begin{equation}
  \tilde T^t_i=\max(\tfrac12\kappa (i-\varphi^t)+T^{t+1}_i,-\tfrac12\kappa (i-\varphi^t)+T^{t-1}_i),
\end{equation}
where $\kappa=\min(1,\omega)$ and $\varphi^t=\max(1,\omega)t+\phi$. After this sequence of dressing transformations, {using \eqref{eq:ud H sub}, the resulting $T$-function then yields an explicit expression} for the solution $U^t_i$ of udKdV, corresponding to the given initial data.

In Figure~\ref{fig:simu} we show this solution at every tenth time step in a moving frame of speed 1. Solitons of mass not exceeding 1 and background are stationary in this frame. The plots show both the time evolution of the initial state calculated using the update rule \eqref{eq:update rule}, plotted as dots, as well as the exact solution found by the sequence of dressing Darboux transformations described above, with $i$ taken to be a real variable, plotted as a continuous line. We observe that, as they should, the two plots coincide at integer values of $i$. One observes the emergence from the initial conditions of three mass 3 solitons (very quickly) and a mass $\frac43$ soliton (around $t=34$). Also, the state translating at speed 1 (and hence stationary in the moving frame) is a mass $\frac12$ soliton copropagating with the background, as is clear from state (e) in the undressing.

\begin{figure}[htbp]\centering
\includegraphics*[width=14.5cm]{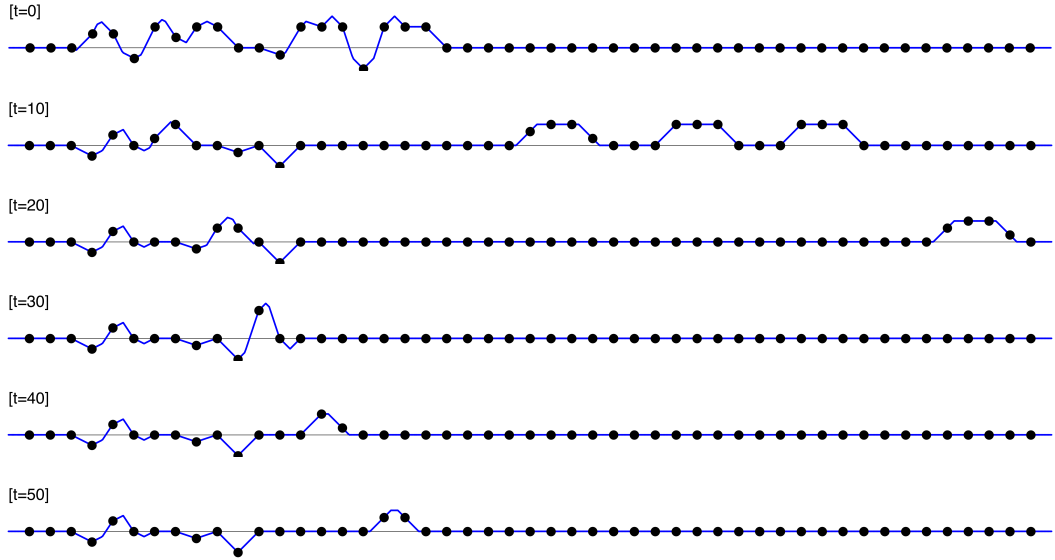}\vskip-.15cm
\caption{\label{fig:simu} Simulation (dots) and exact solution (solid line) shown together. The solution is plotted at time steps $t=0,10,\dots,50$ in a frame moving to the right at speed 1 {(the left-most dot in each plot is located at $i=t+1$).}}
\end{figure}

\end{exa}

\section{Conclusions} 
{We have given an explicit description of eigenfunctions for the $\max$-linear  system for the udKdV equation and we have shown how the soliton adding (dressing) and soliton removing (undressing) procedures for the udKdV equation (over $\mathbb R$, for solutions with compact support) defined by these eigenfunctions, may be performed explicitly, and in complete generality. These processes were shown to be ultradiscete analogues of the Darboux transformation for the dKdV equation.}

 As part of proving these results we discovered two new conserved densities, {$\max X_i^t$ and $\max Y^t_i$, with} $X_i^t$ and $Y^t_i$ defined by \eqref{eq:X} and \eqref{eq:Y}. {These new conserved densities} correspond to conservation of maximal soliton mass and speed respectively.

\section{Acknowledgements}
J{JC}N would like to acknowledge partial support from the Edinburgh Mathematical Society Research Fund and from the Glasgow Mathematical Journal Trust.

RW would like to acknowledge support from the Japan Society for the Promotion of Science (JSPS), through the JSPS grant: KAKENHI grant number 15K04893. He would also like to express his gratitude for financial support (from EMS and GMJT) during a visit to Glasgow in spring 2016, {during which this work took shape}.

\appendix
\section{A proof of Lemma~\ref{lem:m}}\label{ap:proof of 9}
\paragraph{Proof of part (i)} There is a local maximum at $i=m$ if and only if $X_{m-1}^t\le X_m^t$ and $X_{m}^t\ge X_{m+1}^t$, i.e.\ $\Delta X^t_{m-1}\ge0$ and $\Delta X^t_{m}\le0$. Hence Lemma~\ref{lem:m}(i) states that 
\begin{gather*}
(\Delta X^t_{m-1}\ge0\text{ and \/}\Delta X^t_{m}\le0\text{ and \/}X^t_m>1)\\
\implies((U^t_{m-1}+U^t_{m}\ge1\text{ or \/}U^{t-1}_m+U^t_{m}=1)\text{ and \/} (U^t_{m-1}+U^t_{m}\ge1\text{ or \/}U^t_{m-1}+U^{t+1}_{m-1}=1)).
\end{gather*}
The contrapositive version of this implication is
\begin{gather*}
((U^t_{m-1}+U^t_{m}<1\text{ and \/}U^{t-1}_m+U^t_{m}<1)\text{ or \/} (U^t_{m-1}+U^t_{m}<1\text{ and \/}U^t_{m-1}+U^{t+1}_{m-1}<1))\\
\implies(\Delta X^t_{m-1}<0\text{ or \/}\Delta X^t_{m}>0\text{ or \/}X^t_m\le1).
\end{gather*}
It is sufficient to prove the stronger results:
\begin{equation}\label{eq:lem1}
(U^t_{m-1}+U^t_{m}<1\text{ and \/}U^{t-1}_m+U^t_{m}<1)\implies (\Delta X^t_{m-1}<0\text{ or \/} X^t_m\le 1),
\end{equation}
and
\begin{equation}\label{eq:lem2}
(U^t_{m-1}+U^t_{m}<1\text{ and \/}U^{t}_{m-1}+U^{t+1}_{m-1}<1)\implies (\Delta X^t_{m}>0\text{ or \/} X^t_m\le 1).
\end{equation}

Let $U^t_{m-1}+U^t_{m}<1$. Then, using \eqref{eq:naive}, we get 
\begin{equation}\label{eq:naive eq 1}
-U^t_{m}=\max(U^{t}_{m-1}-1,-U^{t}_{m})=\max(U^{t-1}_m-1,-U^{t-1}_{m-1})\ge -U^{t-1}_{m-1},\text{ i.e.\ $U^{t}_m\le U^{t-1}_{m-1}$}
\end{equation}
and
\begin{equation}\label{eq:naive eq 2}
-U^t_{m-1}=\max(U^{t}_{m}-1,-U^{t}_{m-1})=\max(U^{t+1}_{m-1}-1,-U^{t+1}_{m})\ge -U^{t+1}_{m},\text{ i.e.\ $U^{t+1}_m\ge U^{t}_{m-1}$}.
\end{equation}
Then, using \eqref{eq:basic 2}, these give
\begin{equation}\label{eq:Y ineq}
  \Delta Y^t_{m-1}\ge0\text{\ \ \ and \ \ }\Delta Y^{t+1}_{m-1}\le0.
\end{equation} 

Now suppose that, in addition, $U^{t-1}_m+U^t_{m}<1$. By \eqref{eq:basic 4}, 
\begin{equation}\label{eq:Y 0}
Y^t_m=0,  
\end{equation}
and so $\Delta Y^t_{m-1}=-Y^t_{m-1}\le0$. Taken together with \eqref{eq:Y ineq} we get $\Delta Y^t_{m-1}=0$ and so by \eqref{eq:basic 2},
\begin{equation}\label{eq:U=U}
U^t_m=U^{t-1}_{m-1}.
\end{equation}

Also from \eqref{eq:Y ineq}, $\Delta Y^{t+1}_{m-1}\le0$ and so $ Y^{t+1}_{m-1}\ge Y^{t+1}_{m}\ge0$. There are two possibilities: either $Y^{t+1}_{m-1}=0$, and so $Y^{t+1}_{m}=0$ or $Y^{t+1}_{m-1}>0$. In the former case, since from \eqref{eq:Y 0} $Y^t_m=0$ also and \eqref{eq:basic 3} gives $X^t_m=U^t_m+U^{t+1}_m\le1$. In the latter case, \eqref{eq:basic 4} gives $U^{t+1}_{m-1}+U^{t}_{m-1}=1$ and so from \eqref{eq:basic 1},
\begin{equation}
\Delta X^t_{m-1}=U^{t-1}_{m-1}-U^{t+1}_{m-1}=U^{t-1}_{m-1}+U^{t}_{m-1}-1=U^{t}_{m}+U^{t}_{m-1}-1<0,
\end{equation}
using \eqref{eq:U=U}. This completes the proof of \eqref{eq:lem1}. 

The proof of \eqref{eq:lem2} is very similar. Now suppose that, rather than $U^{t-1}_m+U^t_{m}<1$, we have $U^{t}_{m-1}+U^{t+1}_{m-1}<1$. By \eqref{eq:basic 4}, 
\begin{equation}\label{eq:Y 0'}
Y^{t+1}_{m-1}=0,  
\end{equation}
and so $\Delta Y^{t+1}_{m-1}=Y^{t+1}_{m}\ge0$. Taken together with \eqref{eq:Y ineq} we get $\Delta Y^{t+1}_{m-1}=0$ and so by \eqref{eq:basic 2},
\begin{equation}\label{eq:U=U'}
U^{t+1}_m=U^{t}_{m-1}.
\end{equation}

Also from \eqref{eq:Y ineq}, $\Delta Y^{t}_{m-1}\ge0$ and so $ Y^{t}_{m}\ge Y^{t}_{m-1}\ge0$. There are two possibilities: either $Y^{t}_{m}=0$, and so $Y^{t}_{m-1}=0$ or $Y^{t}_{m}>0$. In the former case, since from \eqref{eq:Y 0'} $Y^{t+1}_{m-1}=0$ also and \eqref{eq:basic 3} gives $X^t_m=U^{t-1}_{m-1}+U^{t}_{m-1}\le1$. In the latter case, \eqref{eq:basic 4} gives $U^{t}_{m}+U^{t-1}_{m}=1$ and so from \eqref{eq:basic 1},
\begin{equation}
\Delta X^t_{m}=U^{t-1}_{m}-U^{t}_{m-1}=1- U^t_m-U^t_{m-1}>0,
\end{equation}
using \eqref{eq:U=U'}. This completes the proof of \eqref{eq:lem2}. 

\paragraph{Proof of part (ii)} In the case that $V^t\le1$ Lemma~\ref{lem:V^t} gives $Y^t_i=0$ for all $i$. Then, by \eqref{eq:basic 3}, 
\begin{equation}
\kmax=\omax=X^t_m=U^{t-1}_{m-1}+U^t_{m-1},
\end{equation}
as required.

For $V^t>1$, $\kmax=1$ and we use a proof by contradiction. Suppose that $U^{t-1}_{m-1}+U^{t}_{m-1}<1$. Then, by \eqref{eq:basic 4}, $Y^{t}_{m-1}=0$. In general $Y^{t+1}_{m-1}\ge0$ and we must distinguish two cases. In the case, $Y^{t+1}_{m-1}=0$,  
\begin{equation}
X^t_m=U^{t-1}_{m-1}+U^{t}_{m-1}<1,
\end{equation} 
using \eqref{eq:basic 3}, and so $X^t_m$ is not (globally) maximal. On the other hand, if $Y^{t+1}_{m-1}>0$, \eqref{eq:basic 4} gives $U^{t}_{m-1}+U^{t+1}_{m-1}=1$ and then by \eqref{eq:basic 1},
\begin{equation}
  \Delta X^t_{m-1}=U^{t-1}_{m-1}-U^{t+1}_{m-1}=U^{t-1}_{m-1}+U^{t}_{m-1}-1<0,
\end{equation}
by assumption, which means that $X^t_m$ is not maximal.

\section{Proof of Lemma \ref{lem:samesplit}}\label{proofLemmasamesplit}
Let $m^t$ be the left-most split point in the block $M_j^t$ that corresponds to the soliton that was taken out from $U_i^t$ in the undressing with the bound state eigenfunction $\bar\Theta_i^t$. Since in that case $X_{m^t}^t=\omax$ and $\Delta X_{m^t-1}^t>0$, we have from Proposition \ref{lem:m-1} (ii) that $X_{m^t-1}^{t-1}=\omax$ as well.

Let $\widehat F_i^t$ denote the difference between the two basic solutions to the linear system \eqref{eq:lin1}--\eqref{eq:lin4} for $\widehat U_i^t$ (as in \eqref{eq:F}) with $\omega=\omax$ and $\phi=\phi_{\max}$:
\begin{equation}
\widehat F^t_i=\kappa_{\max}(i-\varphi_{\max}^t)+\sum_{j\ge i}\widehat U^{t}_j-\sum_{j<i}\widehat U^{t-1}_j.
\end{equation}
It then suffices to show that $\widehat F_i^t$ takes its smallest non-negative value at $i=m^t$ (cf. the proof of Theorem \ref{thm:gen sol}).

Since $m^t$ is one of the split points that give rise to the undressing \eqref{UtoUhat}, we have $\widehat U_i^t=U_i^{t-1}$ for $i\geq m^t$, and therefore
\begin{align}
\widehat F_{m^t}^t &= \kappa_{\max}(i-\varphi_{\max}^t)+\sum_{j\ge m^t} U^{t-1}_j-\sum_{j<m^t}\widehat U^{t-1}_j\nonumber\\
&= \sum_{j<m^t} \big( U_j^t - \widehat U_j^{t-1} \big),\label{fhatsums}
\end{align}
where we have used the definition \eqref{eq:varphi} for $\varphi_{\max}^t$. We will now show that $\widehat F_{m^t-1}^t<0$ but $\widehat F_{m^t}^t\geq 0$ which, since the sequence $\widehat F_i^t$ is non-decreasing in $i$ (cf. Theorem \ref{thm:gen sol}), tells us that $\widehat F_{m^t}^t$ is indeed the smallest non-negative value in the sequence $\widehat F_i^t$.

From \eqref{eq:F inc} we have
\begin{align}
\widehat F_{m^t-1}^t &= \widehat F_{m^t}^t + \widehat U_{m^t-1}^t + \widehat U_{m^t-1}^{t-1} - \kappa_{\max}\nonumber\\
&= \sum_{j<m^t-1} \big(  U_j^t - \widehat U_j^{j-1} \big) + \big( U_{m^t-1}^t + U _{m^t-1}^{t+1} - \kappa_{\max}\big),
\end{align}
since $\widehat U_{m^t-1}^t = U_{m^t-1}^{t+1}$ because of \eqref{UtoUhat}.
Furthermore, because of our choice of $m^t$ for which $\Delta X_{m^t-1}^t>0$ and because of Proposition \ref{prop:X ineq}, it turns out that $U_{m^t-1}^t + U _{m^t-1}^{t+1} - \kappa_{\max}<0$ and therefore that
\begin{equation}
\widehat F_{m^t-1}^t < \sum_{j<m^t-1} \big(  U_j^t - \widehat U_j^{t-1} \big).\label{fhatdownsums}
\end{equation}
Note that the $\widehat U_i^{t-1} $  in expressions \eqref{fhatsums} and \eqref{fhatdownsums}  are the result of undressing the downdate $U_i^{t-1}$ of $U_i^t$ with the downdate $\bar \Theta_i^{t-1}$ of the bound state eigenfunction $\bar \Theta_i^t$, for an appropriate split point $m^{t-1}$. Under our assumptions on $m^t$ we have that $X_{m^t-1}^{t-1}=\omax$ and that $m^t-1$ is indeed one of the split points for $\bar \Theta_i^{t-1}$, but it remains to be ascertained that this split point is part of the correct block $M_j^{t-1}$, or in other words that $\bar \Theta_i^{t-1}$ with the choice $m^{t-1}=m^t-1$, is indeed the downdate of $\bar \Theta_i^{t}$ (for split point $m^t$).

Using expression \eqref{eq:bar Theta alt} for the bound state eigenfunction $\bar \Theta_i^{t-1}$ for $U_j^{t-1}$, with $m^{t-1}=m^t-1$, we have
\begin{equation}\label{eq:bar Theta down}
\bar \Theta^{t-1}_i=
\begin{cases}
\displaystyle\kmax(i-m^t +1)+ \sum_{j=i}^{m^t-2} U^{t-1}_j + \sum_{j< m^t-1} U^{t-2}_j&i\le m^t-1\\
\displaystyle\sum_{j<i} U^{t-2}_j&i\ge m^t-1,
\end{cases}
\end{equation}
If we now follow the recipe explained in Section \ref{sec:BSeig} for obtaining the update of this bound state eigenfunction, we find that
\begin{equation}\label{eq:bar Theta downup}
\bar \Theta^{t-1}_i\Big|_{t\to t+1}=
\begin{cases}
\displaystyle\kmax(i-m^t +1)+ \sum_{j=i}^{m^t-2} U^{t}_j + \sum_{j< m^t-1} U^{t-1}_j&i\le m^t-1\\
\displaystyle\sum_{j<i} U^{t-1}_j&i> m^t-1,
\end{cases}
\end{equation}
where, although both cases coincide for $i=m^t-1$, we chose to write the second case with a sharp inequality. As $\kappa_{\max}=U_{m^t-1}^t+ U_{m^t-1}^{t-1}$ because of Lemma  \ref{lem:m} (ii), it is clear that this expression coincides with formula  \eqref{eq:bar Theta alt} for $\bar\Theta_i^t$ and thus that the pair $\bar \Theta_i^{t-1}$ (with split point $m^t-1$) and  $\bar \Theta_i^{t}$ satisfies the linear system \eqref{eq:lin1}--\eqref{eq:lin4} for $U_i^{t-1}$ (Theorem \ref{th:BS}). Then, repeating the calculation that led to the undressing formula \eqref{eq:undress U} but now for $U_i^{t-1}$, with $\bar\Theta_i^{t-1}$ and $\bar\Theta_i^t$, we obtain 
\begin{equation}
  U^{t-1}_i+\bar\Theta^{t-1}_{i+1} + \bar\Theta^{t}_i -\bar\Theta^{t-1}_i-\bar\Theta^{t}_{i+1}=
  \begin{cases}
    U^{t}_i&i <m^t-1\\
    U^{t-2}_i&i\ge m^t-1,
  \end{cases}
\end{equation}
which is nothing but the downdate of $\widehat U_i^t$ \eqref{eq:undress U} with split point $m^{t-1}=m^t-1$.
For this split point, $\bar \Theta_i^{t-1}$ is therefore indeed the downdate of $\bar \Theta_i^{t}$, used in the undressing, and we can now use the fact that $\widehat U_j^{t-1} = U_j^t$ when $j<m^t-1$ to simplify the righthand sides of \eqref{fhatsums} and \eqref{fhatdownsums}. For $\widehat F_{m^t-1}^t $, from \eqref{fhatdownsums}, we clearly have $\widehat F_{m^t-1}^t <0$ and for $\widehat F_{m^t}^t$ we find from \eqref{fhatsums} and \eqref{eq:basic 1} that
\begin{align}
\widehat F_{m^t}^t = U_{m^t-1}^t - \widehat U_{m^t-1}^{t-1} =  U_{m^t-1}^t - U_{m^t-1}^{t-2} = -\Delta X_{m^t-1}^{t-1} \geq 0,
\end{align}
since $X_{m^t-1}^{t-1}$ is maximal and therefore $\Delta X_{m^t-1}^{t-1} = X_{m^t}^{t-1} - X_{m^t-1}^{t-1} \leq 0$.
This completes the proof.

 \end{document}